\documentclass[12pt]{article}
\usepackage[utf8]{inputenc}

\usepackage[notes=true,later=false,camera=false]{dtrt}
\usepackage[utf8]{inputenc}
\usepackage{ stmaryrd }
\usepackage{xspace,enumerate}
\usepackage[T1]{fontenc}
\usepackage[full]{textcomp}
\usepackage[american]{babel}
\usepackage{mathtools}
\hypersetup{
    colorlinks=true,
    linkcolor=blue,
    citecolor=blue,
    urlcolor=blue
}
\usepackage{amsthm}
\usepackage{empheq}
\usepackage{svgcolor}
\def\showauthornotes{1}
\ifnum\showauthornotes=1
\newcommand{\pnote}[1]{\textcolor{red}{ {\textbf{(Pravesh: #1)}}}}
\newcommand{\jnote}[1]{\textcolor{blue}{ {\textbf{(Jeff: #1)}}}}
\else
\newcommand{\pnote}[1]{}
\newcommand{\jnote}[1]{}
\fi

\def \N {\mathbb{N}}
\def \R {\mathbb{R}}

\def \E {\mathbb{E}}

\def \eps {\epsilon}
\def \al {\alpha}
\renewcommand{\Pr}{\mathop{\bf Pr\/}}

\newcommand{\Var}{\mathrm{Var}}

 \def\1{\bm{1}}

\usepackage{amsmath,amssymb,amsthm,amsfonts,latexsym,bm,bbm,xspace,graphicx,float,mathtools,mathdots,physics}
\usepackage{braket,caption,subcaption,ellipsis,xcolor,textcomp,hhline,pifont,combelow,booktabs}
 \usepackage{color}
\usepackage{times}
\usepackage{fullpage}
\usepackage{tikz-cd}
\usepackage[shortlabels]{enumitem}
\usepackage{thm-restate}
\usepackage{cleveref}
\usepackage{mdframed}

\newtheorem{theorem}{Theorem}[section]
\newtheorem{lemma}[theorem]{Lemma}
\newtheorem{claim}[theorem]{Claim}
\newtheorem{proposition}[theorem]{Proposition}
\newtheorem{fact}[theorem]{Fact}
\newtheorem{corollary}[theorem]{Corollary}
\newtheorem{conjecture}[theorem]{Conjecture}
\newtheorem{question}[theorem]{Question}

\newtheorem{definition}[theorem]{Definition}
\newtheorem{remark}[theorem]{Remark}

\newtheorem{observation}[theorem]{Observation}

\newtheorem{examples}[theorem]{Example}

\newcommand{\polylog}{\operatorname{polylog}}

\DeclareMathAlphabet{\mathsfit}{\encodingdefault}{\sfdefault}{m}{sl}
\SetMathAlphabet{\mathsfit}{bold}{\encodingdefault}{\sfdefault}{bx}{n}

\mathchardef\mhyphen="2D
\newcommand{\val}{\text{val}}

\newcommand{\mul}{\text{mul}}

\renewcommand{\int }{\text{Int}}

\usepackage{newpxtext} 
\usepackage{textcomp} 
\usepackage{mathpazo}
\usepackage[scr=rsfso]{mathalfa}
\usepackage{bm} 
\let\mathbb\varmathbb
\usepackage{microtype}

\allowdisplaybreaks

\begin{document}
\title{Smooth Trade-off for Tensor PCA via Sharp Bounds for Kikuchi Matrices}
\author{Pravesh K. Kothari \thanks{Princeton University. Supported by an NSF Career Award \# 2047933, an NSF Medium Grant \# 2211971, and a Sloan Research Fellowship.} \and Jeff Xu\thanks{Carnegie Mellon University. Supported by NSF Medium Grant \# 2211971.}}
\maketitle

\thispagestyle{empty}
\setcounter{page}{0}
\thispagestyle{empty}
\setcounter{page}{0}

%

\abstract{
In this work, we revisit algorithms for Tensor PCA: given an  order-$r$ tensor of the form $T = G+\lambda \cdot v^{\otimes r}$ where $G$ is a random symmetric Gaussian tensor with unit variance entries and $v$ is an unknown boolean vector in $\{\pm 1\}^n$, what's the minimum $\lambda$ at which one can distinguish $T$ from a random Gaussian tensor and more generally, recover $v$? As a result of a long line of work, we know that for any $\ell \in \N$, there is a $n^{O(\ell)}$ time algorithm that succeeds when the signal strength $\lambda \gtrsim \sqrt{\log n} \cdot n^{-r/4} \cdot \ell^{1/2-r/4}$. The question of whether the logarithmic factor is necessary turns out to be crucial to understanding whether larger polynomial time allows recovering the signal at a lower signal strength. Such a smooth trade-off is necessary for tensor PCA being a candidate problem for quantum speedups~\cite{SOKB25}. It was first conjectured by~\cite{WAM19} and then, more recently, with an eye on smooth trade-offs, reiterated in a blogpost of Bandeira~\cite{bandeira_kikuchi_blogpost, bandeira2024oberwolfach}.

In this work, we resolve these conjectures and show that spectral algorithms based on the Kikuchi hierarchy \cite{WAM19} succeed whenever $\lambda \geq \Theta_r(1) \cdot n^{-r/4} \cdot \ell^{1/2-r/4}$ where $\Theta_r(1)$ only hides an absolute constant independent of $n$ and $\ell$. A sharp bound such as this was previously known only for $\ell \leq 3r/4$ via non-asymptotic techniques in random matrix theory inspired by free probability~\cite{bandeira2024matrixconcentrationinequalitiesfree}. 

 
 
Our main technical contribution is a new framework for proving spectral norm bounds on Kikuchi matrices that are tight up to an absolute constant. Along the way to our result, we also confirm a suspicion that Kikuchi matrices are, in general, not intrinsically free -- a property necessary for the free probability-inspired techniques ~\cite{bandeira2024matrixconcentrationinequalitiesfree} to work when $\ell$ grows beyond a fixed constant.


\clearpage
\newpage
\section{Introduction} In this work, we study the tensor principal component analysis (PCA) problem. In this setting, the input is a tensor $T$ obtained by adding a ``noise" tensor $G$ to a symmetric rank-$1$ tensor (the ``signal") on a boolean vector $v\in \{\pm 1\}^n$. Every entry of $G$ is an independent centered Gaussian random variable with variance $1$. 
\[
T = G + \lambda \cdot v^{\otimes r}.
\]
The central algorithmic question is to understand the minimum $\lambda$ --- a measure of the \emph{signal strength} --- at which it is possible to distinguish $T$ from the background ``noise" $G$ and, more generally, recover the planted vector $v$. 

The first interesting case of $r=2$ is the classical single-spiked Gaussian PCA problem, which, in addition to being the fundamental question in algorithmic statistics, has also been of great interest in statistical physics. A celebrated result~\cite{baik2005phase} establishes the famous BBP \emph{sharp} phase transition for this model: the planted signal is efficiently recoverable for any $\lambda>\frac{1}{\sqrt{n}} $\footnote{In particular, this threshold is a factor $1/2$ smaller than the spectral norm of the noise matrix $G$, a finding that may be surprising at first.} and informationally-theoretically impossible to recover when $\lambda < \frac{1}{\sqrt{n}} $. This phenomenon is often described as the existence of a \emph{sharp algorithmic threshold} (the minimum strength at which the planted signal can be detected/recovered) for the PCA problem. 

The case of $r>2$ has been extensively studied as a canonical problem in statistical inference, starting with the work of Montanari and Richard~\cite{MontRichard14}. On the one hand, the planted signal is information-theoretically detectable\footnote{We use $\gtrsim$ to denote an inequality that holds up to a multiplier that is an absolute constant independent of $n$ but may depend on $r$.} whenever $\lambda \gtrsim n^{-\frac{r+1}{2}}$. The best known polynomial-time algorithms\cite{HSS15, HSSS16} (which in fact run in \emph{near-linear} time!) , however, are known to succeed only when $\lambda \gtrsim \sqrt{\log n} \cdot n^{-r/4}$. This potential gap (which arises only for $r \geq 3$) between the signal strength required for information-theoretic and efficient detection/recovery is a major topic of study in average-case algorithm design. \cite{perry2016statistical,  Zdeborov__2016, lesieur2017statistical, MNS18, brennan2018reducibility, jagannath2020statistical} 

\paragraph{Can the minimum signal strength be improved?} Whether there are algorithms for tensor PCA that succeed at a lower signal strength has been a central question in average-case algorithm design. Indeed, investigations of this question have already  led to new techniques that have profound consequences for algorithm design and beyond. About a decade ago, researchers~\cite{BhattiproluGuruswamiLee17, BhattiproluGhoshGuruswamiLeeTulsiani17, RRS17}  found higher order spectral methods that run in subexponential-time algorithms and improve the above algorithmic results. In their landmark work, Wein, Alaoui, and Moore~\cite{WAM19} introduced spectral methods based on \emph{Kikuchi Matrices} to substantially simplify algorithms and their analyses. Kikuchi matrices have since had a deep impact on algorithm design and beyond \cite{guruswami2022boolean, hsieh2023hypergraph, alrabiah2023near, kothari2023exponential, kothari2024superpolynomial}.

\begin{theorem}[\cite{BhattiproluGuruswamiLee17, BhattiproluGhoshGuruswamiLeeTulsiani17, RRS17,WAM19}]
For any $1\leq \ell \leq n$, there is a $n^{O(\ell)}$-time algorithm for detection and recovery in tensor PCA model that succeeds with probability tending to $1$ whenever 
\[ 
\lambda  \gtrsim \sqrt{\log n} \cdot n^{-r/4} \cdot \ell^{1/2- r/4}\,,
\] 
for an absolute constant $C>0$.
\end{theorem}  
On the flip side, lower bounds in various restricted models provide evidence that no algorithm could improve the guarantees on $\lambda$ above by more than a polylogarithmic factor\cite{PR22,kunisky2019ldlr,DudejaHsu21,MontRichard14, BenArousGheissariJagannath20}. 

\paragraph{Smooth tradeoffs: can higher polynomial times help?} The best polynomial-time guarantee for tensor PCA is currently achieved by a near-linear time algorithm. Does allowing higher polynomial time help? If true, such a result establishes that there is no sharp algorithmic threshold for tensor PCA for $r>2$ unlike the case of $r=2$ and other well-studied statistical inference problems such as community detection in stochastic block models~\cite{decelle2011asymptotic, abbe2017community, ghosh2020sos_skpaped, wein2021optimal, bandeira2021spectral}. 

While such a smooth trade-off has been suspected by researchers for half a decade now, rigorously establishing it has proved challenging. The key technical difficulty is in the analyses of the spectral norms of highly correlated random matrices that arise in the spectral methods for the problem. Even with the simplifications offered by the Kikuchi matrices, proving sharp bounds has proven difficult. Indeed, in their original work that introduced these matrices, Wein, Alaoui, and Moore~\cite{WAM19} conjectured an improved bound on the minimum signal strength that incurs no logarithmic loss. The significance of this result and its consequences for establishing a smooth trade-off in algorithmic thresholds for tensor PCA was recently highlighted in a conjecture of Bandeira. We state his conjecture below and refer the reader to the accompanying blogpost for a more detailed discussion of the context. 

\begin{conjecture}[\cite{bandeira2024oberwolfach, bandeira_kikuchi_blogpost}] \label{conj:bandeira}
For any even $r$ and $\ell\geq \frac{r}{2} $ and any $n$, let $\lambda_{r,\ell}$ be the threshold given by the spectral norm of the Level-$\ell$ Kikuchi matrix. For any fixed $r$, one has \[ 
n^{r/4} \cdot \lambda_{r,\ell} \rightarrow 0
\]
as $\ell \rightarrow \infty$ (crucially noting that this is after one has taken $n\rightarrow \infty$).
\end{conjecture}

In addition to being a question of central interest in statistical inference, establishing such a smooth trade-off has consequences for establishing potential quantum speedups. Indeed, a recent work~\cite{SOKB25} shows a \emph{quartic} (i.e., a factor four improvement, in contrast to the expected ``Grover-search type" factor two improvement, in the exponent) speedup in estimating the top eigenvector of Kikuchi matrices for tensor PCA over the classical power iteration. Their result suggests tensor PCA as a candidate example for quartic quantum speedups provided one could establish that higher polynomial running times strictly lower the signal strength required for success of the Kikuchi spectral methods.

\paragraph{Spectral Norm Bound for Random Matrices as the Bottleneck.}
 
Let us describe the key technical barrier in establishing the conjecture above in more detail now. At a high-level, the main issue is that the level $\ell$ Kikuchi matrix defined as following has size $\Theta(n^{2\ell})$  but only $O(n^{r})$ independent bits of randomness. When $\ell \gg r$ -- the previously unexplored setting -- this matrix has highly correlated entries. 

\begin{definition}[Kikuchi Matrix for  Even-$r$. See \cref{def: kikuchi-odd} for the definition for odd $r$] \label{def: kikuchi-even}
	For even $r$, and any $\ell \in \N$, and given tensor $G$ of order-$r$, $M_\ell$ is the Kikuchi matrix of size $n^{\binom{n}{\ell}} \times n^{\binom{n}{\ell}} $ for level-$\ell$ indexed by subsets  $I,J \subseteq \binom{n}{\ell}$ with entry $M_\ell(G) [I,J] \coloneqq G_{I\Delta J}  $ for $|I \Delta J| = r$, and $0$ otherwise..
	
	We usually drop the dependence on $G$ when it is clear. For the bulk of our work, unless otherwise specified, we will be working with $G$ being a random symmetric Gaussian tensor of order-$r$.
\end{definition}	
\begin{examples}
	For indices $I= \{v_1,v_2,v_3,v_4,v_5,v_6\}, J=\{v_1,v_2,v_3,v_4, v_7,v_8\}$, in the case of $r=4$ and $\ell = 6$, we have the corresponding entry of the Kikuchi matrix be \[ 
	M_\ell [I,J ] \coloneqq   G_{I\Delta J} = G_{v_5, v_6, v_7, v_8}\,.
	\]
\end{examples}


More precisely, the analyses of the Kikuchi spectral method for tensor PCA relates a high probability estimate on the spectral norm of the Kikuchi matrix built from $G$ to the threshold signal strength at which detection is possible as follows:
%
%

\begin{claim}
	For a given upper bound $B(G)$ on a random tensor $G$ that holds with high probability, one has a distinguishing algorithm for \[\lambda \geq  \Theta_r(1) \cdot  \frac{B(G)}{n^{r/2} \cdot \ell^{r/2}}  \,.\]
\end{claim}

In the original line of work establishing the "coarse" thresholds, the spectral norm bounds for Kikuchi matrices all suffer a loss of polylog factors, yielding bounds \[ B_{folklore}(G) = O_r(1)  \cdot n^{r/4}\cdot \ell^{r/4+1/2}\cdot   \sqrt{\log n}\,, \] as an immediate application of black-box matrix concentration bounds (eg. Matrix Bernstein). 

Let us make two comments about the above bound. First, note that the $ \sqrt{\log n}$ factor dominates the effect of increasing $\ell$ so long as it is bounded by an absolute constant. Thus, one cannot rigorously infer a smooth trade-off in polynomial time regime of the algorithm. Secondly, the dependence of $\ell$ should also be tightly controlled apart from the  $ \sqrt{\log n}$ factor - restricted to the dependence of $\ell$, a bound of $o(\ell^{r/2})$ is necessary to establish any trade-off between $\ell$ and $\lambda$. 
%

It is important to point out that controlling the dependence on the two key parameters simulataneously -- $\log n$ and the $\ell$ -- turns out to be difficult. For example, the recent work of d'Orsi and Trevisan~\cite{dorsi_trevisan23} introduced an Ihara–Bass-type formula that yields sharper spectral norm estimates for the “basic” spectral algorithm in the random $k$-xor setting. However, while their bound does get rid of the $\log n$ factor, their analyses, when run for larger $\ell$ has a worse polynomial dependence on $\ell$. In retrospect, their strategy relies on analyzing the non-backtracking walk matrix and is tied to the setting where the input tensor is sparse. In contrast, our input tensor is fully dense.  



\paragraph{Recent Progress and Barriers in Free Probability.} 

 A series of works \cite{Bandeira_2023, bandeira2024matrixconcentrationinequalitiesfree, BLNvR25} introduce techniques from free probability to prove sharp bounds on correlated random matrices via improved noncommutative Khintchine inequalities. These works have used Kikuchi matrices arising in tensor PCA as a testbed for applications. In particular, a recent work  \cite{bandeira2024matrixconcentrationinequalitiesfree} makes non-trivial progress on \cref{conj:bandeira}. 
\begin{theorem}[\cite{bandeira2024matrixconcentrationinequalitiesfree}] For any even $r$ and any $\frac{r}{2}\leq  \ell < \frac{3r}{4}$, let $G$ be a random symmetric tensor of order-$r$ as defined above,  and let $M_\ell(G)$ be the Kikuchi-$\ell$ matrix, with high probability \[ 
\|M_{\ell}(G)\|_{sp} \leq   \underbrace{ O_r(1) \cdot n^{r/4} \ell^{r/4}}_{\coloneqq B_{fp}(G)  } \,. 
\] 
\end{theorem}
For even r (they do not analyze their algorithm for odd $r$ that often tends to be similar but more technically challenging), they prove the conjecture for all $\ell \leq 3r/4$. 

The bound above is based on non-asymptotic techniques in random matrix theory from free probability. In particular, the techniques end up establishing not only a spectral norm bound as above but also a semicircular shape of the spectrum as in Wigner random matrices. This property is called \emph{intrinsic freeness}. In their work\cite{bandeira2024matrixconcentrationinequalitiesfree}, the authors suggest that intrinsic freeness is likely false for the Kikuchi matrices above when $\ell \gg r$. 


\begin{quote}(Remark 3.4) \label{remrk: afonso-free}
 When $\ell$ is large compared to $r$, the dependence structure of $M_\ell$ is so strong that it is unclear whether it could be accurately modeled by (the corresponding free matrix) $M_{\text{free}}$.	
\end{quote}

That said, if the matrix is indeed intrinsically free, its spectrum would have a small deviation to that of $M_{\text{free}}$, predicting a norm bound of $B_{fp}$. As we will see, along the way to our main result, we confirm their suspicion and show that
an additional $\sqrt{\ell}$ factor is necessary via a lower bound for the expected high trace, establishing the the matrix does not exhibit semicircle spectrum. We highlight the lower bound construction below \cref{lem:lower-bound}  with formal verification in \cref{sec:lower-bound}. 



In an incomparable improvement on the above bound, a recent work of Bandeira and Nizic-Nikolac shows the following bound that removes the logarithmic factor for all $\ell$ but has a worse polynomial dependence on $\ell$. As we discussed above, finding a proof strategy that simultaneously eliminates the logarithmic factor while preserving the right dependence on $\ell$ has proven challenging despite significant efforts. 

\begin{theorem}[\cite{bandeira_nizic_kikuchi_pc}] 
In the same set-up as above except now for any $\ell \in \N$,
 \[ 
\|M_\ell(G)\|_{sp} \leq  \underbrace{ O_r(1) \cdot n^{r/4}    \cdot  \ell^{r/2}}_{\coloneqq B_{pc}(G)   }\,.
\]
	
\end{theorem}

To summarize, the first bound $B_{fp}(G)$ establishes a smooth trade-off but only for running time regimes where $r/2 \leq \ell \leq 3r/4$.  The second bound $B_{pc}$, on the other hand, removes the additional constraint on $\ell$ at the cost of a worse dependence on $\ell$ and does not imply a smooth trade-off between running time and signal strength for tensor PCA in polynomial time. 

In this work, we prove a sharp bound on $\|M_{\ell}(G)\|$ (for both even and odd $r$) for all $\ell \ll n^{\Omega(1)}$ and as a special case, resolve \cref{conj:bandeira}.

\subsection{Our Results}
 We are now ready to describe our main result. Firstly, we state our main result on sharp spectral norm bounds for Kikuchi matrices, and defer the formal definition of the odd-$r$ case to the subsequent section.
 
\begin{theorem}[Improved Spectral Norm Bound for Kikuchi Matrix for Even-$r$] \label{thm:main-thm-even} For any even $r>2$,  and $G$ a random symmetric tensor of order-$r$ with entry $G_S \sim N(0,1)$ for any $S\in \binom{n}{r}$, let $M_\ell(G)$ be the level-$\ell$ Kikuchi matrix of $G$ defined in \cref{def: kikuchi-even},
there exists some constant $\delta>0$ such that for any $\ell \in \N$ such that $\ell <n^\delta $, with probability at least $1-o_n(1)$, \[ 
\|M_\ell(G)\|_{sp} \leq   O_r(1) \cdot \left(\sqrt{n \cdot \ell }\right)^{r/2} \cdot \sqrt{\ell }\,.
\]
where $O_r(1)$ hides constants independent of $n$ and $\ell$ but possibly depending on $r$.
\end{theorem}

As noted before, the prior bound of\cite{bandeira2024matrixconcentrationinequalitiesfree} worked only when $\ell \leq 3r/4$ and was established only for even $r$. The case of odd $r$ has been challenging in all analyses of Kikuchi matrix method~\cite{guruswami2022boolean, hsieh2023hypergraph, alrabiah2023near, kothari2023exponential, kothari2024superpolynomial}. 

\begin{theorem}[Improved Spectral Norm Bound for Kikuchi Matrix for Odd-$r$]\label{thm:main-thm-odd}  For any odd $r>2$,  and $G$ a random symmetric tensor of order-$r$ in the above normalization,  let $M_\ell(G)$ be the level-$\ell$ Kikuchi matrix of $G$ defined in \cref{def: kikuchi-odd}, there exists some constant $\delta>0$ such that for any $\ell \in \N$ such that $\ell \leq  n^\delta$, with probability at least $1-o_n(1)$, 
\[ 
\|M_\ell(G)\|_{sp} \leq O_r(1) \cdot \left(\sqrt{n\cdot \ell}\right)^r \,.\]
\end{theorem}

\begin{remark}
	 Our bound works to $\ell = n^\delta$ for some constant $\delta>0$ and gives stronger concentration bound that holds with error probability at most $\exp(n^\delta)$ as opposed to the polynomial tail bound from black-box non-commutative Khintchine.
\end{remark} 

 
 From our improved norm bounds for Kikuchi matrices, we show that the Kikuchi hierarchy gives a smooth trade-off for the Tensor PCA problem in the full polynomial-time regime, proving the conjecture of \cite{bandeira_kikuchi_blogpost}.

\begin{theorem}[Smooth Tradeoff for Tensor PCA] \label{thm:tensor-pca-application}  There exists some constant $\delta$ such that for any $\ell\leq n^\delta$,
	there is an algorithm with runtime $n^{O(\ell)}$ that solves the detection problem of Tensor PCA with probability at least $1-o_n(1)$ for \[
	\lambda > C_r \cdot n^{-r/4 } \cdot \ell^{1/2-r/4}
	 \]
	 for some absolute constant $C_r$ that depends solely on $r$. 
	 
	 Concretely, there is an algorithm $f(T): \R^{n^r} \rightarrow \{0, 1\} $ running in time $n^{O(\ell)}$ that satisfies $ 
	 \Pr_G[f(G)  = 1 ] = o_n(1) 
	 $ if $\lambda = 0$ and $G$ is a random symmetric tensor, while   $ 
	 \Pr_{G, v}[f( G +\lambda \cdot v^{\otimes r} )  = 1 ] = 1-o_n(1) 
	 $ if $\lambda >   C_r \cdot n^{-r/4 } \cdot \ell^{1/2-r/4} $ and $v$ is a uniformly chosen boolean vector $\{\pm 1\}^n$.
\end{theorem}

\paragraph{Downstream Algorithmic Application From Detection.}
Our sharp threshold result for the detection problem also gives rise to the following extensions by combining with the known techniques from \cite{WAM19}: 1) extension to more general prior; and 2) a recovery algorithm for the same regime.

\begin{theorem}[Extension to General Prior]
	Similar to Thmeorem B.1 in  \cite{WAM19}, our result extends to the general prior $v$ with i.i.d. entries drawn from some distribution $\pi $  on $\R$ (that does not depend on $n$) and normalized such that $\E[\pi^2]=1$ with some constant overhead in the range of the signal-noise ratio $\lambda$ as above. \end{theorem}
\begin{remark}
Though not explicitly carried out, our technique extends to settings in which the underlying random tensor of i.i.d. entries with mean-$0$ and variance-$1$ have weak-yet-polynomial dependence in $n$ in the higher-order moments in a fashion similar to norm bounds for adjacency matrix of sparse random graphs of at least polylogarithmic average degree \cite{xu2024switchinggraphmatrixnorm} .
\end{remark}

\begin{theorem}[Application for Recovery ] \label{thm:recovery} For  $r>2$, and $\ell \in N$,
	let $v$ be the planted signal in Tensor PCA instance as defined above. Let $c>0$ be some absolute constant, and $C_r>0$ be some constant depending solely on $r$ from our norm bounds. There is an algorithm that outputs a vector $\widetilde{v}$ in time $n^{O(\ell)}$  for all $\eps>0, \delta\in (0,1)$,  and $\lambda\geq c \cdot C_r \cdot \epsilon^{-4} \cdot  n^{-r/4 } \cdot \ell^{1/2-r/4}$ such that \[ 
	corr(v,\widetilde{v}) \geq 1- c\cdot   \eps
	\]   
	with probability $1-\delta$ 
	provided $\ell \leq n\eps^2$.
\end{theorem}

\paragraph{Tightness of Our Bounds: Kikuchi Matrix is Not Free for $\ell\gg r$.} 

Besides the algorithmic applications, we also give a lower bound for the trace calculation to demonstrate the tightness of our spectral norm bounds.  
\begin{lemma} For any even $r>2$,  there exists some constant $\delta>0$ such that for any $\ell, q \in \N$ and $\ell, q \leq n^\delta $, and $G$ is a random symmetric tensor of order $r$ with normalization prescribed as above, 
	\[ 
	\E \Tr[ (M_\ell)^{2q}]  \geq   \binom{n}{\ell}\cdot \left( \Theta_r(1)  \cdot  \sqrt{n^{r/2} \cdot \ell^{r/2}} \cdot \sqrt{\ell }\right)^{2q} \,.
	\]
	with probability at least $1-o_n(1)$ where $\Theta_r(1) $ hides constants solely depending on $r$ but independent of $n, q$ and $\ell$.
\end{lemma}

\begin{remark}
	 Since we allow $\Theta_r(1)$ slack in the above, this does not pose a conflict for the bound $B_{fp}(G)$ for the regime $\ell < r$. Moreover, 
	  this bound is most interesting for first picking $r$ and then choosing $\ell > C \cdot r$ for some large constant $C$.
\end{remark}

As a by-product of our techniques, by ruling out the semicircle spectrum, this lower bound confirms the suggestion of remark 3.4 in \cite{bandeira2024matrixconcentrationinequalitiesfree}  that Kikuchi matrix of larger-$\ell$ is not intrinsically free. To highlight this discrepancy, observe that $B_{fp}$ does not contain the extra factor of $\sqrt{\ell}$ in our final norm bound, while as we will show in the technical sections, this extra factor in fact admits an intuitive combinatorial interpretation. Moreover, our technique also explains why the additional factor of $\sqrt{\ell}$ is absent for the small-$\ell$ regime of $\ell <r$. 

Furthermore, we showcase that in the case the underlying tensor drawn from i.i.d. Rademacher distribution, the above lower bound on the expected trace is sufficient to imply a lower tail bound for spectral norm.
\begin{lemma} For any even $r$, and for $G$ a random symmetric tensor with each entry being i.i.d. from $G_S\sim \{\pm 1\}$ for $S\in \binom{n}{r}$, with probability at least $1-o_n(1)$, \[
\|M\|_{sp} \geq \Theta_r(1) \cdot \sqrt{n^{r/2} \cdot \ell^{r/2}} \cdot \sqrt{\ell }\,.
 \]
	
\end{lemma}

	\paragraph{Organization of Our Work} In the next section, we give a technical overview of our techniques by featuring the case of even-$r$, and we present the formal proof in the section 3. We present the lower bounds in section 4.

\begin{remark}
    Between the submission of our manuscript and its posting on arXiv, we were made aware of a concurrent work by ~\cite{li2025smoothcomputationaltransitiontensor} that solves the distinguishing variant of the question. We would like to point out that our result for refutation also implies a similar one for the distinguishing question.
\end{remark}

\section{Technical Overview}
In this section, we start with some technical preliminaries and set the stage for the second half of this section that serves as a technical overview of our work. We will be primarily focusing on the even-$r$ case to showcase our techniques. As Kikuchi matrices are the primary focus of our study, let's begin by formally introducing them. To make our technical analysis intuitive, we will also introduce the diagram perspective of the Kikuchi matrices.
\subsection{Kikuchi Matrices: Definitions and Diagrams.} For completeness, we reiterate the definition for Kikuchi matrix for even-$r$. 
\begin{definition}[Kikuchi Matrix for  Even-$r$] 
	For even $r$, and any $\ell \in \N$, and given tensor $G$ of order-$r$, $M_\ell$ is the Kikuchi matrix of size $n^{\binom{n}{\ell}} \times n^{\binom{n}{\ell}} $ for level-$\ell$ indexed by subsets  $I,J \subseteq \binom{n}{\ell}$ with entry $M_\ell(G) [I,J] \coloneqq G_{I\Delta J}  $ for $|I \Delta J| = r$, and $0$ otherwise..
	
	We usually drop the dependence on $G$ when it is clear. For the bulk of our work, unless otherwise specified, we will be working with $G$ being a random symmetric Gaussian tensor of order-$r$.
\end{definition}	

For odd $r$, albeit less straightforwardly, one may still define Kikuchi matrix in the following way as inspired by the "Cauchy-Schwarz" trick that is standard in this line of works.
\begin{definition}[Kikuchi Matrix for Odd-$r$] \label{def: kikuchi-odd}
	For odd $r$, and any $\ell \in \N$, and given tensor $G$ of order-$r$, we let $M_\ell$ denote the corresponding Kikuchi matrix indexed by subsets $I,J \subseteq \binom{n}{\ell}$ with entry \[ 
	M_\ell [I,J ] = \sum_{t\in [n]\setminus (S\cup T)} \sum_{\substack{S_1, S_2 \\ \text{partition of } S\setminus T \\ |S_1|=|S_2| } }\sum_{\substack{T_1, T_2 \\ \text{partition of } T\setminus S \\ |T_1|=|T_2| } } G_{S_1 \cup T_1 \cup \{t\}} \cdot  G_{S_2 \cup T_2 \cup \{t\}} 
	\] for $|I\Delta J| = 2r-1$, and $ 0$ otherwise.  
\end{definition}
As a sanity check, observe that in the above definition, each matrix entry is a degree-$2$ polynomial of the underlying input. To facilitate our technical analysis, we will borrow the diagram representations from the study of \emph{Graph Matrices} - a family of random matrices with entries being polynomial functions of the underlying input.  We refer interested readers to \cite{AMP20, HKPX23, KPX24} for a more involved introduction, while we briefly review some terminologies in minimality  necessary for the discussion of our work. 


 In short, one can view the two oval on the left and right as row and column index for the corresponding matrix entry, and each circle inside the oval as an index in $[n]$. In this walk, we will refer to them as \emph{step-boundaries} in the walk as they can be viewed as start or destination of some step in the walk.  Moreover, we represent the underlying (hyper-)edge \footnote{With some abuse of notation, we will not distinguish between hyperedge and edge in our technical analysis.} (random variable) via the orange rectangle (in the case $r=4$), and the dotted line connecting two vertices in the row and column index that share the same index in $[n]$, i.e. any gray vertex is a vertex from the intersection of the row and column indices. 
%
\begin{figure}[h!]
\centering
\begin{subfigure}{.5\textwidth}
  \centering
  \includegraphics[width=.8\linewidth]{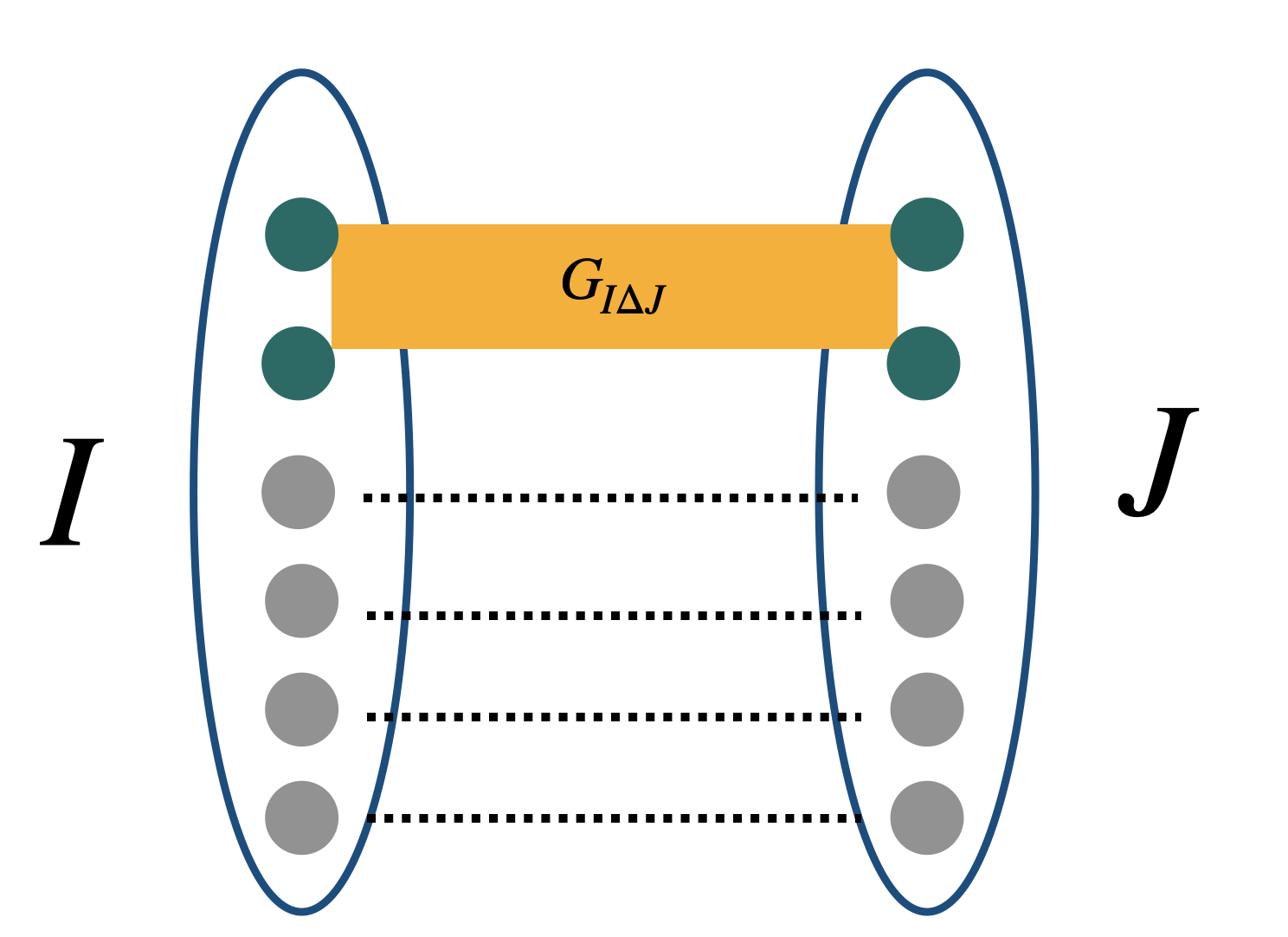}
  \caption{Kikuchi Matrix for $(r=4, \ell = 6)$}
  \label{fig:sub1}
\end{subfigure}%
\begin{subfigure}{.5\textwidth}
  \centering
  \includegraphics[width=.8\linewidth]{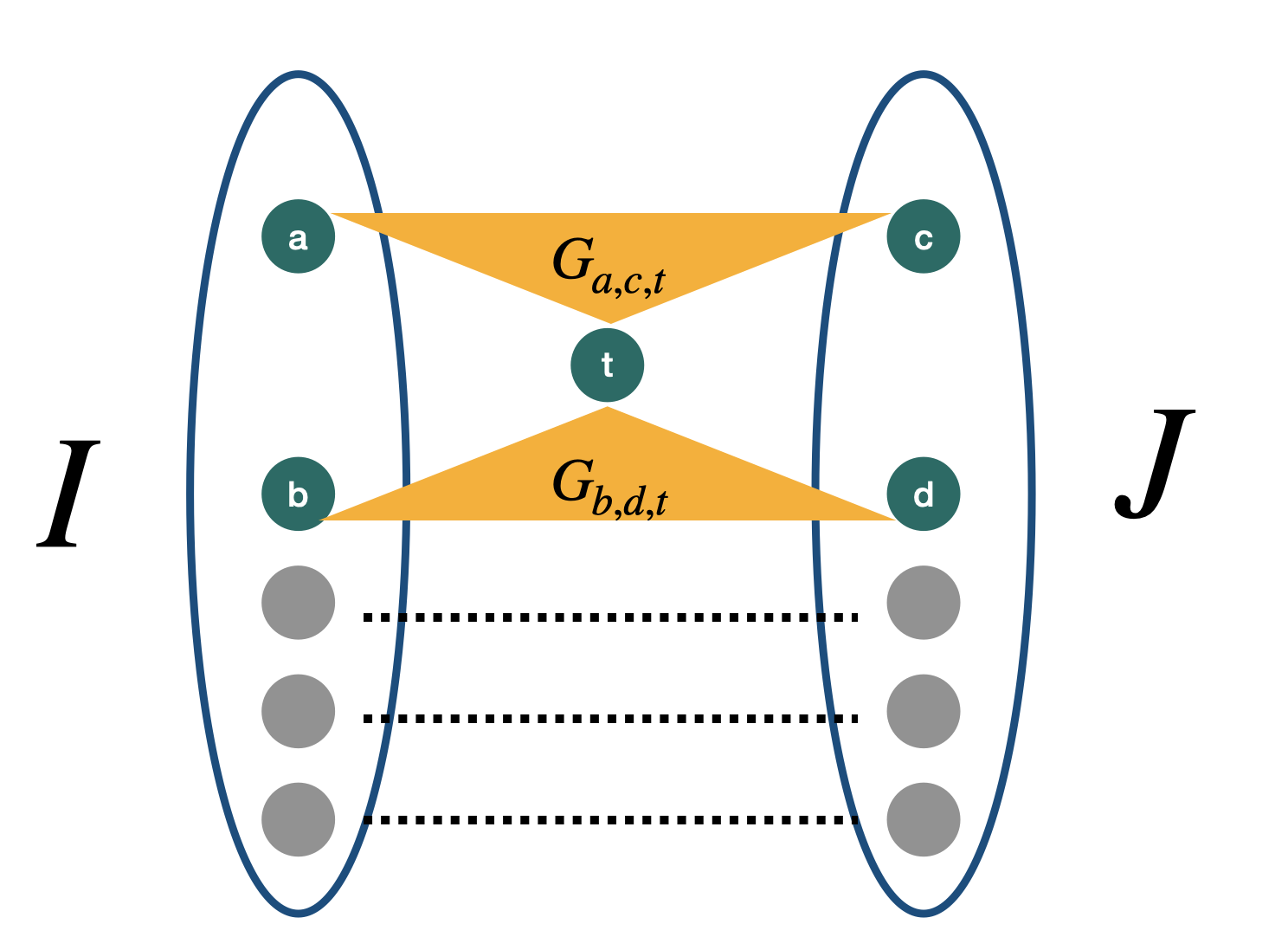}
  \caption{Kikuchi Matrix for $(r=3, \ell = 5)$}
  \label{fig:sub2}
\end{subfigure}
\end{figure}

\subsection{Trace Moment Method and Spectral Norm Bounds}
Our starting point is the trace moment method, the classic technique from random matrix theory for understanding the spectrum of random matrices. 
Despite being a rather elementary technique as we shall see, it has been particularly successful in yielding tight asymptotic bounds, usually as the first successful approach, for random matrices. Its versatility and effectiveness span across classical ensembles such as Gaussian Orthogonal Ensembles and (centered) adjacency matrices of sparse random graphs \cite{Fri08, https://doi.org/10.1002/rsa.20406, BLM15, specgapdensegraph, Bor19} , as well as structured random matrices with intricate dependencies and significant entry-wise correlations that arise in modern algorithmic applications \cite{HSSS16, RRS17, Bhattiprolu2017SumofSquaresCF, DMOSS19, AMP20, MOP19ExplicitRamanujan, JPRTX, KPX24, HKPX23, xu2024switchinggraphmatrixnorm}.

 Trace moment method generally proceeds as following, for any symmetric matrix $M$, and $q\in N$, \[
\|M\|_{sp}^q \leq  \sum_{i} \lambda_i(M)^{2q} =   \Tr(M^{2q}) \,.
 \]
When $M$ is a random matrix, it suffices for one to bound $
\E[\Tr(M)^{2q}]$ by taking a sufficiently high power of the matrix with $q =\omega(\log n)$.  At this point, one can observe that the RHS $\Tr(M^{2q})$ admits a purely combinatorial interpretation, as \[ 
 \Tr(M^{2q}) = \sum_{\substack{P:\text{closed-walk of length-2q} \\P= (v_1,v_2,...,v_{2q}=v_1)} } \prod_{i\in [2q]} M[v_i, v_{i+1}]\,.
\]
This is simply the sum of (weighted) closed walks of length-$2q$ on matrix $M$.  Specializing to our setting of Kikuchi matrix for even-$r$ and taking expectation on both sides, we have \[ 
 \E_G \Tr(M_\ell^{2q}) = \sum_{\substack{P:\text{closed-walk of length-2q} \\P= (S_1,S_2,...,S_{2q}=v_1) \\ S_i\in \binom{n}{\ell} }  } \E_G \left[ \prod_{i\in [2q]} M_\ell [S_i, S_{i+1}] \right]\,.
\]
Since each hyper-edge corresponds to a Gaussian random variable with odd-moment being $0$, one can also infer that each hyper-edge must appear for an even number of times in the walk ( i.e. used by an even number of steps).
Therefore, any closed walk with non-zero contribution in the summand admits the following combinatorial description.

\begin{definition}[Trace-walk for even-$r$] \label{def:simple-trace-walk-even}
	A trace-walk $P$ for Kikuchi matrix of level-$\ell$ at length-$2q+1$ is a sequence of $\ell$-sized subsets  $S =  \{S_t\}_{t\in [2q+1]} $  such that \begin{enumerate}
		\item $S_i \in \binom{n}{\ell}$ for any $i\in [2q]$;
		\item $|S_i \Delta S_{i+1}| = r$ for any $i\in [2q]$, and the $i$-th step is along the hyperedge (equivalently viewed as the random variable) $G_{S_i \Delta S_{i+1}}$ ;
		\item (Closed-walk) $S_1 = S_{2q+1}$;
		\item (Evenness) Each hyper-edge appears in an even number of steps. 
	\end{enumerate}
	Equivalently, we view it as a walk of length-$(2q)$ such that at each step-$t$, we move from the left boundary $U_t = S_t$ to the right boundary $V_{t} = S_{t+1} = U_{t+1}$.
\end{definition}

For intuition, the following example is a snippet of a trace-walk restricted to the first $3$ steps for $r=4$. This example also highlights a property of working with matrix indices being subsets as opposed to order-tuples-adjacent steps may not be incident to any common vertex. It prompts us to make the following definition for the each step of the walk.

 \begin{figure}[h!]
\centering
\centerline{\includegraphics[totalheight=6cm]{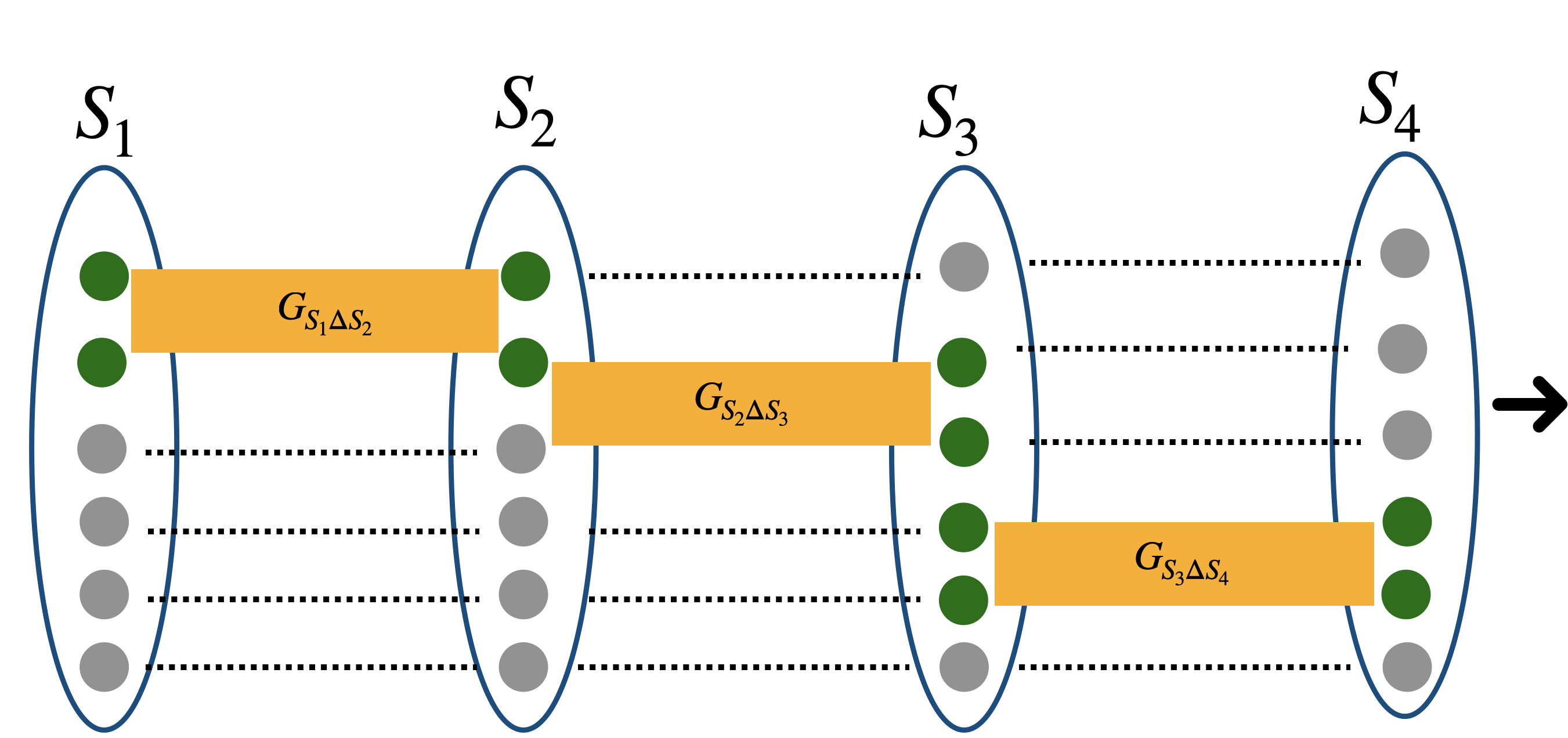}}
    \caption{Snippet of a Trace-Walk}
    \label{fig:shape-walk1}
\end{figure}

\begin{definition}[Vertex Status for a Given Step]
	For any step in the trace-walk in which we move from $U_t = S_t \in \binom{n}{\ell}$ to $V_{t} = S_{t+1} \in \binom{n}{\ell}$, we call \begin{enumerate}
		\item each vertex in $S_{t+1} \cap S_t = U_t\cap V_t$ \emph{dormant};
		\item each vertex in $S_{t+1}\Delta S_t = U_t\Delta V_t$ \emph{active}.
	\end{enumerate}
	In particular, we additionally call each active vertex in $S_{t+1}\setminus S_t$ \emph{incoming-active}, and each vertex $S_{t}\setminus S_{t+1}$ \emph{outgoing-active}. In short, we also refer to each such vertex as "incoming" and "outgoing". 
\end{definition}

In the above diagram, each gray vertex is a "dormant" vertex in both of the adjacent steps, while a vertex is green if it is "active" in either of the adjacent steps.

Additionally, for each walk $P$, we additionally refer to its corresponding expected value as  walk-value as the following,
\begin{definition}[Walk-value]
	For each walk $P = \{S_1, S_2, \dots, S_{2q-1}, S_{2q}=S_1\}$, we define its walk-value as \[ 
	\val(P) \coloneqq \E_G \left[ \prod_{i\in [2q]} M_\ell [S_i, S_{i+1}] \right] = \prod_{e = (I,J)\in E(P)} \E_G[(G_e)^{\mul_P(e)} ]  \,.
	\] 
\end{definition}

And a crucial observation at this point is that there are two kinds of factors one needs to control when trying to upper bound the expected trace, \begin{enumerate}
	\item Combinatorial (Vertex) Factor: this keeps track of the factor that arises from counting the walks (i.e., the cost of identifying each vertex's label in $[n]$ throughout the trace-walk).
	\item Analytical (Edge-value) Factor: this corresponds to the walk-value, i.e., the expectation of random variables traversed through the walk. 
\end{enumerate}

\paragraph{Tight Bounds from Local Assignment Schemes.} Traditionally, trace moment analyses typically rely on global structural constraints of the contributing walks. For instance, in the case of Gaussian random matrices, the walks that contribute most to the trace moment correspond exactly to Dyck paths, where half the steps visit “new” vertices and the other half return to previously visited ones. While effective in simple settings, this kind of global reasoning becomes cumbersome when dealing with more complex graph matrices, and one major source of complication and intimidation is arguably the requirement to keep track of long walks. Hypothetically, the analysis would feel much more intuitive if, like the case of Frobenius norm, it suffices work with length-$2$ walks as opposed to typically polylog or even longer walks that are usually required for spectral norm bounds.

This insight lies at the heart of recent advances that aim to streamline trace moment calculation, making it more intuitive to use, more amenable to complex dependencies, and still capable of delivering sharp bounds. Put simply, the core idea of the developed scheme is to break the correlations across long walks by using a factor-assignment method that distributes the global counting factors to individual steps. This enables a local, step-by-step analysis without needing to track the complex dependencies spanning the entire walk.
 That said, most of the technical challenge of the calculation arises in designing a valid factor-assignment scheme that strikes a delicate balance: \begin{enumerate}
	\item Completeness: it must account for all global contributions;
	\item Sharpness: the induced local bound for each step must match candidate global upper bounds.
\end{enumerate}

 We are now ready to introduce the up-shot of a given factor-assignment scheme,
\begin{definition}[Factor-assignment scheme and Step-Bound Function $B_q(\cdot)$] \label{def:factor-assignment-scheme}
	For any matrix $M$ and $q \in \N$, for any factor assignment scheme, we
call $B_q(M)$ a valid step-bound function of the given scheme if
\[ 
\E[\Tr(M^{2q})] \leq \mathsf{MatrixDimension} \cdot B_q(M)^{2q}
\]
and for each step of the walk, the factor-assignment schemes assigns to the individual step a factor at most $ B_q(M)$. The dependence on $q$ and $M$ is usually dropped when clear. 
\label{def:step-bound-function}
\end{definition}

Once such a scheme is specified, finding a spectral norm bound is then fairly straightforward: it reduces to verifying that each step of the walk individually satisfies the corresponding "global" upper bound $B_q(\cdot)$, which is simply the candidate spectral norm bound one is shooting for. This local approach extends naturally to more complex matrix structures and has proven quite effective in prior work when the underlying shape corresponds to a graph in various settings\cite{JPRTX, KPX24, HKPX23, xu2024switchinggraphmatrixnorm} . 
 
 However, the success of the this approach so far has been restricted to random matrices with underling input being a matrix or graph , and does not apply well to tensor inputs. It is not clear whether the particular approach would yield interesting bounds for Kikuchi matrices with underlying input being a tensor or a hypergraph, as the design of the local factor-assignment scheme in the previous works is highly combinatorial and graph-theoretic. As we shall see, the fine-grained techniques is indeed sensitive to the change of the underlying input from matrix to tensor,  and unique phenomena pop up in the case of tensor input!

\paragraph{Recovering the Folklore Bounds for Rademacher Tensor.}

To shed light on the local factor-assignment scheme that we will be working with, we start with a lightweight application in this fashion that allows us to recover the vanilla bound of $B_{folkdlore}(G) =  O_r(1)  \cdot n^{r/4}\cdot \ell^{r/4+1/2}\cdot   \sqrt{\log n} $. Undoubtedly, this bound can be readily obtained by a straightforward application of anyone's favorite matrix concentration bound, while we resort to the heavy machinery of the trace moment method with the hope that we may identify what needs to be improved when we seek a tighter bound that shaves off the log factor.

\begin{claim} \label{claim:naive-bound-log}
	For any even-$r>3$ and $\ell\in \N$ , and $G$ a random symmetric tensor such that $G_S \sim \{\pm 1\}$ i.i.d., let $M\coloneqq M_\ell(G)$ be the associated Kikuchi matrix, and any $q = \Theta(\ell \log n )$, with probability at least $1-o_n(1)$,  \[ 
	\|M\|_{sp} \leq  O(1) \cdot  \sqrt{\binom{n-\ell }{r/2} \cdot \binom{\ell}{r/2} }  \cdot \sqrt{\ell } \cdot \sqrt{\log n}\,.
	\]  
	
\end{claim}

	Before we describe the precise factor-assignment scheme, we make the following observation. Since we are working with Rademacher random variables as opposed to Gaussians, one convenient simplification is that we can effective ignore the analytical factor that comes from the walk-value: for any walk $P$ in which each hyperedge is traversed an even number of times, we have \[\E_G [\val(P)] = \E_G[ \prod_{e \in E(P)} (G_e)^{mul_P(e)}] =  1\,,
	\]
	and $0$ otherwise where $mul_P(e)$ is the number of times edge $e$ appears walk $P$, i.e., the number of times the walk traverses along the edge $e$. From now on, it suffices for us to focus on the combinatorial factor that arises from walk-counting.
	
	We next observe the following. Given any index $S\in \binom{n}{\ell}$, it has at most $\binom{n-\ell}{r/2} \cdot \binom{\ell}{r/2} $ neighbors in the Kikuchi matrices. Therefore, at time-$t$ in the walk in which we move from $S_t$ to $S_{t+1}$, we have at most $\binom{n-\ell}{r/2} \cdot \binom{\ell}{r/2} $  choices for identifying the hyper-edge, and equivalently, the destination $S_{t+1}$. Moreover, in the case that the hyper-edge has been traversed in some previous step of the walk, such an (hyper-)edge can be more efficiently specified at a cost of at most $2q$. This culminates in the following bound that accounts for the "global" contribution of a hyper-edge to the combinatorial counting,
	
	\begin{proposition}
		In the global combinatorial factor of the walk, for a hyper-edge  $e$ appearing in the walk for $mul_P(e)$ times, it  incurs a total cost of at most $B_{global}(G_e) \leq \binom{n-\ell}{r/2} \cdot \binom{\ell}{r/2}  \cdot (2q)^{mul_P(e)-1}$.
	\end{proposition}
		
	With the above observation, we are now ready to describe the factor-assignment scheme that assigns the factors of each hyper-edge to individual steps that traverse along the particular hyper-edge.
\begin{mdframed}[frametitle = (Combinatorial) Factor Assignment Scheme for Hyper-edges]
Each hyper-edge appears at least twice in the walk, and the first time it appears, a cost of $\binom{n-\ell}{r/2} \cdot \binom{\ell}{r/2} $ is sufficient, and a cost of $2q$ is sufficient for any subsequent appearance. 	\begin{enumerate}
		\item Assign a factor of $\sqrt{\binom{n-\ell}{r/2} \cdot \binom{\ell}{r/2} } \cdot \sqrt{2q}$ for the first and the last\footnote{equivalently, the final} step in which the  hyper-edge appears in the walk;
	\item Assign a factor of $2q$ for any step in which the edge is making a middle (non-first/last) appearance.
\end{enumerate}
\end{mdframed}	

We start by showing that this is a complete scheme, and each (combinatorial counting) factor in the global bound is accounted for via the local assignment. \begin{proposition}
	For any hyper-edge $e$ that appears in the walk, let $B_{local}(e)$ be the total factor assigned \emph{by the above scheme} to this hyper-edge  $e$ throughout the walk across its various appearances, we have 
	\[
	B_{local}(e) \geq \ B_{global}(e)	\,. \]
	\end{proposition}
\begin{proof}
	By our assignment-scehme, we have \[ 
	B_{local}(e) = \left(\underbrace{\sqrt{\binom{n-\ell}{r/2} \cdot \binom{\ell}{r/2} \cdot 2q}}_{\text{First} } \cdot  \underbrace{\sqrt{\binom{n-\ell}{r/2} \cdot \binom{\ell}{r/2} \cdot 2q }}_{\text{Last}} \right)\cdot \underbrace{ (2q)^{mul_P(e)-2}}_{\text{ Middle } }
	\,.\]
	Since each hyper-edge appears at least twice in a contributing walk, we have $\mul_P(e)\geq 2$. And notice that in the base case of $mul_P(e)=2$, the inequality holds as $B_{global}(e) \leq \binom{n-\ell}{r/2} \cdot \binom{\ell}{r/2} \cdot (2q)$.
Moreover, any increment in $mul_P(e)$ gives a factor of $2q$ to the both sides. This completes the proof to our proposition. 
\end{proof}
This then allows us to complete the naive bound with $\sqrt{\ell \log n}$-factor for Kikuchi matrix of random Rademacher tensor.

\begin{proof}[Proof to \cref{claim:naive-bound-log} ]
Since each step uses exactly $1$ hyper-edge, for $\al\in \{F,H, L\}$, let $B_q(\al)$ be the local bound for each individual step using a hyper-edge that appears for first (F), middle (H), and last\footnote{equivalently, the final time} time (L). By construction of the above scheme, we have \[ 
B_q(F) = B_q(L) =  \sqrt{\binom{n-\ell}{r/2} \cdot \binom{\ell}{r/2} \cdot 2q }\,,
\] 
and \[  
B_q(H) = 2q = o_n(1) \cdot  (B(F) + B(L))
\,.\]
provided $q \ll \binom{n-\ell}{r/2} \cdot \binom{\ell}{r/2}$. 

Summing over all possible choices for a given step gives us a local bound of \[
B_q(M_\ell) \leq B_q(F)+B_q(H)+B_q(L) =  (2+o_n(1))  \sqrt{\binom{n-\ell}{r/2} \cdot \binom{\ell}{r/2} \cdot 2q }\,.
 \]
 
 By construction of our factor assignment scheme, this yields \[ 
 \E[\Tr(M_\ell^{2q})] \leq n^\ell \cdot (B_q)^{2q} \,.
 \]
 Finally, plugging the chosen value of $q= \Theta(\ell \cdot \log n) $,
  taking $1/2q$-th root of the above and applying Markov's then gives us a norm bound that holds with probability $1-o_n(1) $ of \begin{align*}
 \|M_\ell\|_{sp} \leq (1 + \eps ) \cdot B_q  &= O(1) \cdot \sqrt{\binom{n-\ell}{r/2} \cdot \binom{\ell}{r/2} \cdot 2q } \\&= O(1) \sqrt{\binom{n-\ell}{r/2} \cdot \binom{\ell}{r/2} }\cdot \sqrt{\ell} \cdot \sqrt{\log n} 	\,.
  \end{align*} 
   \end{proof}

Let us now pause to examine a key bottleneck in the above scheme: where exactly does the extra
$\polylog$ factor arise in the global accounting? The culprit is the factor of 
$2q$ incurred when an edge appears for the second (and the last) time. For ease of the discussion, we call it the \emph{last-cost}.
 To understand why it suffices to focus on cases where the second and the last appearances coincide, observe that any other non-first appearance is treated as a middle appearance in the scheme—and such steps are assigned only negligible local contributions, recalling $B(H)=o(1) \cdot B(F)$ (and analogously $B(L)$).
  
For readers familiar with the trace moment calculation for G.O.E. matrix, it is well-anticipated that the last-cost can be improved: in fact, it may be just $1$ as in the case of G.O.E. or graph adjacency matrix. This is almost correct except that is is also slightly over-optimistic, and it turns out that we are in an intriguing in-between regime:  the cost can be improved while the last-cost would not just be $1$ either! 

\paragraph{Identifying the Target Upper Bound from a Lower Bound.}

To start with, notice that suppose last-cost is $O_r(1)$,  we would obtain final bound of $\|M_\ell\|_{sp}\leq  O_r(1)\cdot  \sqrt{\binom{n-\ell}{r/2} \cdot \binom{\ell}{r/2} } $.   To rule out this hypothetical bound in the framework of trace moment method, it is instructive to take a detour into the  lower bound and consider the corresponding lower bound for the expected high trace. 
\begin{lemma}[Lower Bounding the High Trace] \label{lem:lower-bound}  A factor of $\Theta_r( \ell )$ is necessary for the Last-Cost. As a result, 
\[ 
\E \Tr(M_\ell^{2q}) \geq \textsf{Matrix-Dimension} \cdot \left(  \Theta_r(1) \cdot \sqrt{\binom{n-\ell}{r/2} \binom{\ell}{r/2  } } \cdot \sqrt{\ell } \right)^{2q } 
\]	
\end{lemma}

 \begin{figure}[h!]
\centering
\centerline{\includegraphics[totalheight=6cm]{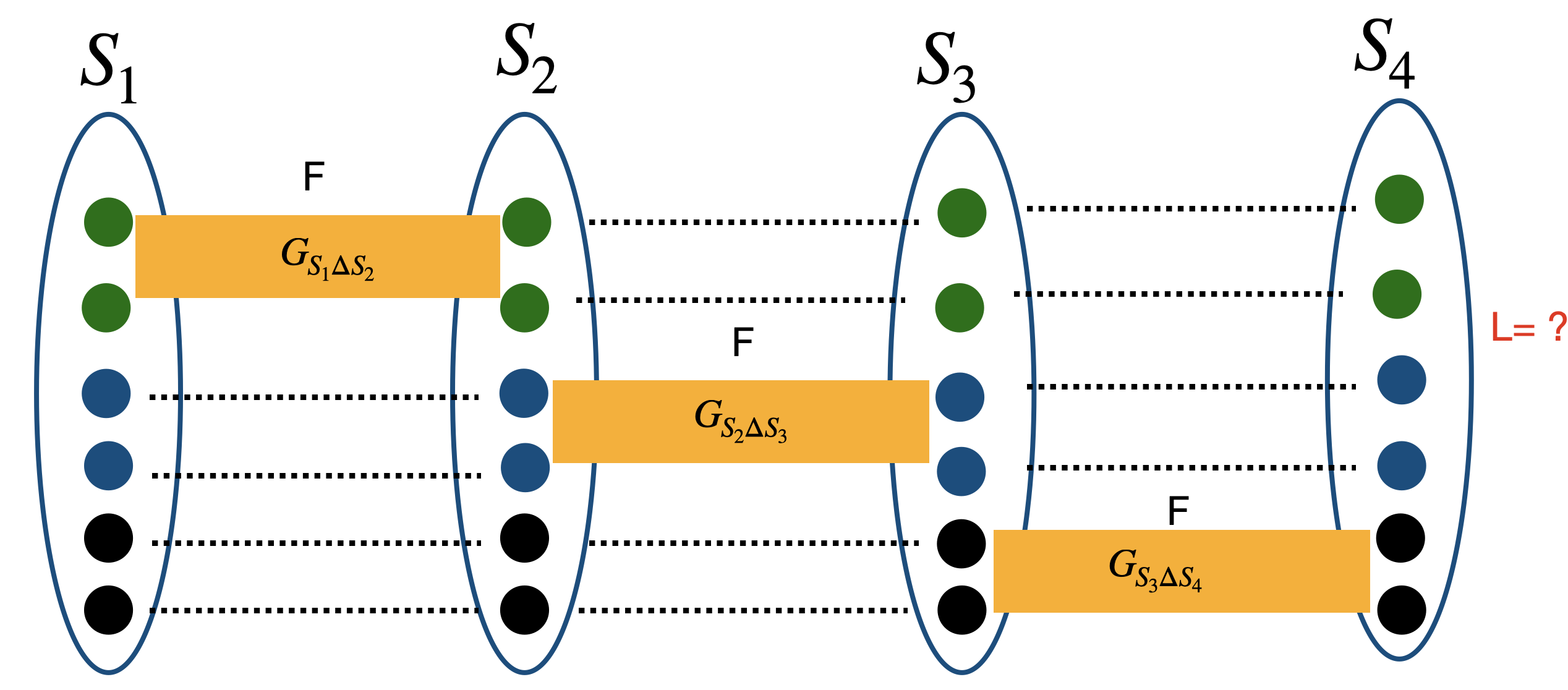}}
    \caption{Confusion about the Upcoming $L$-Step}
    \label{fig:return_confusion}
\end{figure}

We now give a diagram explanation for the extra factor  $\ell$ in the last-cost. For simplicity, assume $r$ is even and $r/2$ divides $\ell$. 
  Lett us consider the first $3$-steps of some walk as described in the above diagram, \begin{question}
	Suppose that we are at the index $S_4$ and we are to use an edge for the second (and last) time in the next step, which edge could we be using?  Equivalently, what could $S_5$ be if this is an $L$-step?
\end{question}
	It's fairly straightforward to see that in the above example, any edge among the first three edges can be making a second (and last) appearance in the next step. And more generally, for larger $\ell$, this reveals that we would have a last-cost of $\Theta(\frac{\ell}{r/2}  ) = \Theta_r(\ell )$ at some particular step.	Moreover, this bound in fact holds for a typical step along edges appearing for the second (and the last) time. Observe that we have a total last-cost for $t\coloneqq \frac{\ell}{r/2}$ edges being $\frac{\ell}{r/2}!$ if we arrange them in the manner of 1) first using $\frac{\ell}{r/2}$ new edges, and 2) subsequently re-use all of the $\frac{\ell}{r/2}$ new edges to return to the start yet in a potentially different ordering. By Stirling approximation, this gives an individual bound of$
	(t!)^{1/t} = \Theta(1)  \cdot t = \Theta_r(\ell)
	$. We defer the formal verification of the lower bound to ~\cref{sec:lower-bound} .

Zooming back, the lower bound has identified a clear target for the upper bound, that is showing a cost of $\ell$ for each step we are using an edge appearing for the second-and-last time in the walk.

\paragraph{When is $\ell$ Sufficient for last-cost?} It'd be straightforward if one can establish a worst-case bound of last-cost being $\ell$ for any step of any walk. However, that is false. However, it turns out that at least in the ideal situation below, one can show that $\ell$ is indeed sufficient.

\begin{definition}[Ideal Last-Step]
	Suppose at some step-$t$ in the walk, we are walking from $S_t$ to $S_{t+1}$ along some edge that is promised to be making its last-appearance in the walk, we call the step from $S_t$ to $S_{t+1}$ an ideal last-step if \emph{some out-going vertex in $S_{t}\setminus S_{t+1} \subseteq [n]$ is appearing for the final time in the walk, i.e., it is not to-be-revisited upon the departure from $S_t$}.
\end{definition} 

Let's now see why this is sufficient to guarantee a cost $\ell$ for specifying the destination of $S_{t+1}$ from $S_t$. Note that throughout the walk, we may can easily keep track of which of the "seen" edges so far in the walk have made their final appearances by incurring a $O(1)$ bound per-step, 
\begin{observation}
	As the walk proceeds, it is a cost of $2$ per-step to identify whether the edge being traveled along at the given step is appearing for the final time.
\end{observation}
Therefore, in an ideal last-step, it suffices for us to identify one of the outgoing vertices from $S_t$ that is not-to-be revisited as guaranteed by our assumption. This incurs a cost of $\ell$ since $S_t$ is an $\ell$-sized subset. Once such a vertex is identified, the edge used in the next step to move to $S_{t+1}$ is also identified since any such vertex must be incident to a unique edge that has not made their final appearance.


We have now shown that the cost of $\ell$ is sufficient under the ideal condition. However, a crucial question remains: why is it sufficient to bound the contribution in the ideal case alone? Recall that, in the worst case, the length of the walk can lead to a bound of $2q$, which we aim to avoid. 

\paragraph{Final Strategy.}

We begin by revisiting the factor-assignment scheme used in the vanilla analysis, which assigns a $2q$-factor to steps along edges appearing for the middle time. However, such middle-time edges can be effectively ignored, as they contribute only lower-order terms compared to edges appearing for the first or final time. This observation motivates our final strategy: to design a new factor-assignment scheme that penalizes last-steps failing to meet the ideal condition. To this end, we go beyond the earlier edge-based approach—which uniformly assigned factors to all $L$-steps regardless of the underlying edge—and instead decompose the edge-based cost into vertex-level contributions. This shift to a vertex-based analysis provides the flexibility to prioritize vertices appearing for the final time, guiding the scheme to favor steps aligned with the ideal condition. We formally introduce this new scheme in the next section.%

\section{Sharp  Norm Bounds from Factor Assignment Scheme}
\subsection{Preliminaries for Trace-Walk}
To facilitate the discussion, we begin by reintroducing several key definitions. While some of these were previously presented in the technical overview for the special case of even $r$, we now restate them in a more general and formal form to ensure that our arguments can be swiftly extended to the case of odd $r$ as well. For convenience, we will highlight the distinction each time a definition is reintroduced and generalized.

For starters, let's reintroduce the definition of trace-walks which would immediately  apply to trace walks of Kikuchi matrices for odd-$r$ as well to incorporate vertices at each step of the walk getting summed over while outside the left/right boundaries. For readability purpose, we suggest the reader to focus on the even-$r$ case in the first pass and simply adhere to the simplified definition introduced in \cref{def:simple-trace-walk-even}.

\begin{definition}[Generalized Definitions of Trace-walk] For any $r>2$,
	A trace-walk for Kikuchi matrix of level-$\ell$ at length-$2q$ is a sequence of sets $(S_1, W_1, S_2, W_2, \dots,S_{2q}, W_{2q},S_{2q+1})$ such that   \begin{enumerate}
		\item Walk-boundary sets: $S_i \in \binom{n}{\ell}$ for any $i\in [2q+1]$;
		\item Intermediate sets : $W_i \subseteq [n]$ for any $i\in [2q]$, and additionally, we assume $W_i$ to come with an two equal-partitions $U_A, U_B$ for $S_i\setminus S_{i+1}$ and $V_A, V_B$ for $S_i\setminus S_{i+1}$;		
		\item  For even-$r$,  $|S_i \Delta S_{i+1}| = r$ for any $i\in [2q]$  and $W_i = \emptyset$ for any $i\in [2q]$;
		\item For odd-$r$, $|S_i \Delta S_{i+1}| = 2r-1$, and $|W_i|=1$ for any $i\in [2q]$; 
		\item (Closed-walk) $S_1 = S_{2q+1}$;
		\item  (Evenness) Each hyper-edge appears in even number of steps. 
	\end{enumerate}
	Equivalently, we view it as a walk of length-$(2q)$ such that
	\begin{enumerate}
		\item  at each step-$t \in [2q]$, we move from the left boundary $U_t = S_t$ to the right boundary $V_{t} = S_{t+1} = U_{t+1}$ with potential intermediate-set $W_t$;
		\item For even-$r$, each step travers along exactly $1$ hyper-edge given by $S_t\Delta S_{t+1} $;
		\item for the odd-$r$, each step travers along $2$ hyper-edges given by  $U_A \cup V_A \cup W_t$ and $U_B\cup V_B\cup W_t$ where we assume that $U_A, U_B$ are equal partitions for $S_{t}\setminus S_{t+1}$ and $V_A, V_B$ for $S_{t+1}\setminus S_t$ as prescribed by the intermediate set.  
        \item Each hyper-edge appears for an even number of times.
	\end{enumerate} 
\end{definition}

With the generalized definition of trace-walk, we may now relate the expected trace of the Kikuchi matrix to the weighted sum of trace-walks as defined above.
\begin{proposition} For any $r\geq 3, \ell \in \N $ and any $q\in \N$, 
	\[ 
	\E[\Tr(M_\ell)^{2q}]  \leq \sum_{P:\text{trace-walk of length }2q} \val(P) \,.
	\]
    where we recall that \[ 
    \val(P) \coloneqq \E_G \left[ \prod_{i\in [2q]} M_\ell [S_i, S_{i+1}] \right] = \prod_{e = (I,J)\in E(P)} \E_G[(G_e)^{\mul_P(e)} ]
    \]
\end{proposition}

We consider the following factor-assignment scheme. For starters, we focus on the "combinatorial" factor that arises from counting of the walk, and a crucial observation to get us started is to split the cost of specifying random variables (viewed as hyper-edges) among vertices, in particular, among the incoming active vertices.

\begin{definition}[Generalized Vertex Status]
	For any step in the trace-walk in which we move from $U_t = S_t \in \binom{n}{\ell}$ to $V_{t} = S_{t+1} \in \binom{n}{\ell}$ potentially through intermediate vertex $W_t$, we call \begin{enumerate}
		\item each vertex in $S_{t+1} \cap S_t = U_t\cap V_t$ \emph{dormant};
		\item each vertex in $(S_{t+1}\Delta S_t) \cup W_t = (U_t\Delta V_t) \cup W_t$ \emph{active}.
	\end{enumerate}
	In particular, we call each active vertex outside $S_t$ \emph{incoming-active} (i.e. vertices from $(W_t\cup V_t)\setminus U_t$ , and each active vertex from $S_{t}$ \emph{outgoing-active} (i.e., vertices from $U_t\setminus V_t$). In short, we also refer to each such vertex as "incoming" and "outgoing". 
\end{definition}
As a remark,  each vertex from the intermediate set $W_t$ is considered active at the corresponding step by the above definition. With this definition in hand,  one can observe that each random variable (viewed as a hyper-edge) only depends on the active-vertices while independent of the dormant vertices. And this is precisely why we call them "active" and "dormant".

\begin{definition}[Vertex-appearance at a given step]
For any step-$t$ in a trace-walk in which we move from $U_t \coloneqq S_t \in \binom{n}{\ell}$ to $V_t \coloneqq S_{t+1} \in \binom{n}{\ell}$, we call a vertex in $S_t \cup S_{t+1} \cup W_t $ making an \emph{appearance} at step-$t$. Moreover,\begin{enumerate}
	\item a vertex $v\in S_t\cup S_{t+1} \cup W_t $ with label in $[n]$ is making a \emph{first} appearance at step-$t$ if there is no $t'<t$ such that any vertex that makes appearance at step-$t'$ has the same label in $[n]$;
	\item  a vertex $v\in S_t\cup S_{t+1}\cup W_t$ with label in $[n]$ is making a \emph{last} appearance if there is no $t'>t$ such that any vertex that makes appearance at step-$t'$ has the same label in $[n]$;
	\item 	 a vertex $v\in S_t\cup S_{t+1}\cup W_t$ with label in $[n]$  is making a \emph{middle} appearance at step-$t$ if it is making neither a first nor middle appearance. 
	\end{enumerate}
	When the step-$t$ is clear, we ignore the dependence of the appearance at $t$ and simply call it a vertex making appearance.
\end{definition}

As a sanity check, we observe that following the above definition, at each step-$t>1$ from $S_t = U_t$ to $S_{t+1} =  V_t$ (possibly through some intermediate set $W_t$, a vertex in the (left) boundary-set $S_t$ makes an appearance in step-$t$ as a vertex from $U_t$, and also an appearance in the previous step-$(t-1)$ as a vertex from $V_{t-1} = S_{t}$. Analogously, vertices in the right-boundary set of step-$t$ $S_{t+1} = V_t$ also make appearances in both step $t$ and step $t+1$.

\begin{proposition} By paying a one-time cost of $n^\ell$ throughout the walk, we can assume the following:
	at each step-$t$ from $S_t=U_t$ to $V_t= S_{t+1}$ , no vertex in $U_t$ is making a first appearance, and no vertex in $V_t$ is making a last appearance.
\end{proposition}

\begin{remark}
	This is immediate for any step except the left-boundary of the first and the right-boundary of the final step of the walk following the above discussion. For the special set $S_1$ being the start and final destination of the walk, we pay a cost of $n^\ell$ up front which would then allow us to assume they are not making "first" appearance in step-$1$ for local accounting purpose (and similarly "last" appearance in step-$2q$).
\end{remark}

Analogously, we extend the definition of "appearance" for vertices to hyper-edges, and thereby steps. The partition of steps according to the step-status will later on be the cornerstone of our local perspective in the factor-assignment scheme. 
\begin{definition}[Step-Status/Edge-Appearance]
	We consider the hyper-edge (or hyoeredges in the case of odd-$r$) being traversed along at step-$t$ as making an appearance at the particular step. Moreover, a hyper-edge is making \emph{first} appearance if it does not appear in previous steps; and analogously it is making a \emph{final} appearance if it does not appear in subsequent steps.

For each step, we define its status based on the edge-appearance of the edge(s) in the current step. In particular, for even-$r$, a step-status is equivalent to the edge-appearance of the edge it travers along, and takes a label in $\{F,H,L\}$. For odd-$r$ in which each step is composed of $2$ edges, a step-status is a label in $\{F,H,L\}\times \{F,H,L\} $.
\end{definition}

With these perliminaries settled, we are ready to describe our final factor-assignment scheme. In the subsequent subsection, we describe the more lightweight component concerning edge-values, and then we describe the most technical component of our work concerning combinatorial-factor in the subsection after.

%
\subsection{Edge-Value Assignment}
Recall that in our toy-example showcased in the previous section for $\{\pm1\}$ random input, each walk gives expected value $0$ or $1$ through the defined quantity of walk-value $\val(P)$ . However, this is no longer true when we consider Gaussian input, more generally, inputs with higher moments. That said, it can be readily handled via an additional component - edge-value assignment - in our factor assignment scheme.  

\begin{fact}
	For $g \in \N(0,1)$ (and $\{\pm1\}$), we have $E[g]=0 $, and $E[g^2] =1$, and moreover, for $t> 2$, \[ 
	|\E[g^t]| \leq (2t)^{t/2}\,.
	\] 
\end{fact}
Given the higher-moment bound for edges, we consider the following edge-value assignment to steps,

\begin{mdframed}[frametitle = Edge -Value Assignment] \label{prop: edge-value-assignment}
Each edge (random variable $g_e$ that appears for $t$ times in the walks contributes an analytical value of $|\E[(g_e)^t]| \leq (2t)^{t/2}$, and we assign this value to each individual step that traverses along edge $e$ via factor $B_{ev}(i,e)$-the factor assigned to step-$i$ from edge-$e$- as following.
\begin{enumerate}
	\item For each step $i$ in which the edge $e$ is making the first or final appearance $(F/L)$ , we assign a factor $B_{ev}(i,e) = 1$;
	\item For each step  $i$ in which the  edge $e$ is making a middle appearance $(H)$, we assign a factor of $B_{ev}(i,e) = 2\cdot 2q = 4q$.
\end{enumerate}	
\end{mdframed}

\begin{definition}[Edge-value Bound for Steps]
For each step-$i$, we define \[B_{ev}(i) \coloneqq \prod_{\substack{ e \\ \text{edge appearing in step-}i  } } B_{ev}(i, e)\]
where $B_{ev}(i,e)$ is the factor assigned to step-$i$ from edge-$e$ in the above scheme.
\end{definition}

We can now verify that all the analytical factors throughout the global walk is accounted locally across steps in the walk from the perspective of each edge $e$ that appears in the walk.
\begin{claim}[Edge-value assignment is complete] For each edge $e \in \binom{n}{r}$, let $\tilde{B}_{ev}(e)$ be the factor assigned to it across its appearances throughout the walk in the above scheme, i.e., $B_{ev}(e) \coloneqq  \prod_{ \substack{ i\in [2q] \\ e \text{ appears in step-}i }} B(i,e) $ ,
we have \[ 
\tilde{B}_{ev}( e) \geq  |\E[(g_e)^t]|
\,.\]
	
\end{claim}
\begin{proof}
	This is immediate for edges that make $2$ appearances, and edges appearing only once (or more generally odd number of times) is zero. For number of appearances $t>2$ with, we have \[ 
	B_{ev}(e)  = (2 \cdot 2q)^{t-2} \geq (2t)^{t/2}
	\]
	as we observe that $t\leq 2q$ trivially, and verify that $x^{t-2} \geq x^{t/2}$ for $t\geq 4$. Finally, the claim is proven as we recall that \[ 
	|\E[g^{t}]| \leq (2t)^{t/2}
	\]
	for $t>2$ for $g\in \N(0,1)$ and $\{\pm 1\}$, and noting that for $t=3$ the RHS is trivially $0$ while the LHS is $2q$.
	\end{proof}
    From the above, 
    it is straightforward to verify that we can upper bound the global walk-value as following 
    once we switch to a step-perspective. This concludes our edge-value component of the factor-assignment scheme.
    \begin{proposition}[Verification of Global Value Bound from Steps]
    For any walk $P$, recall that $\val(P) \coloneqq \E_G \left[ \prod_{i\in [2q]} M_\ell [S_i, S_{i+1}] \right] = \prod_{e = (I,J)\in E(P)} \E_G[(G_e)^{\mul_P(e)} ]$, we have \[ 
    \val(P) \leq \prod_{e\in E(P)} \tilde{B}_{ev}(e)  = \prod_{\substack{ \text{step-}i \\ i\in [2q] \\ }} B_{ev}(i)\,.
    \]
    where we emphasize $\tilde{B}_{ev}(e)$ is the accounting from edge-perspective for edge $e$, and $B_{ev}(i)$ is the accounting from the step-perspective for step-$i$.
\end{proposition}

\subsection{Combinatorial-Factor Assignment via Vertices}
To identify each step, it is sufficient to identify the following for each edge appearing for the first time,\begin{enumerate}
	\item What are the \emph{outgoing} active vertices from the current boundary? This is a cost of $\ell$ for each such vertex. 
	\item What are the \emph{incoming} active vertices? Each of them could take a cost of $[n]$ if they have not been seen, or a cost of $O(q\cdot r)$ otherwise. 
\end{enumerate}
On the other hand, for any edge that has been seen, at most a cost of $O_r(q)$ is sufficient while we note that it can be improved in the ideal-last-step situation to be $O_r(\ell)$. Motivated by the above observation, we design a factor-assignment for each of the above factors. 

\subsubsection{Factor of Outgoing Active Vertices}

For starters, we note that the factor regarding \emph{outgoing} active vertices can be rerouted and split over the edge's first and last appearance by the following scheme,

%

\begin{mdframed}[frametitle = Combinatorial Factor Assignment for Factors of Outgoing Active Vertex in $\ell$ ]
\label{prop:out-going-cost}
For each edge appearing for the first time, it incurs a factor of $\ell^{ \lfloor r/2 \rfloor}$ from identifying the (incidental) outgoing active vertices. We assign them via the following,
\begin{enumerate}
	\item For the step in which the edge first appears ($F$), assign a factor of $\sqrt{\ell}$ for each of its incident \emph{incoming} active vertex on the step-boundary;
	\item  For the step in which the edge appears for the last time  ($L$), assign a factor of $\sqrt{\ell}$ for each of its incident \emph{outgoing} active vertex on the step-boundary;
\end{enumerate} 

\begin{remark}
	Crucially, notice that no $\sqrt{\ell}$ factor is assigned to the intermediate vertex $W$ for the odd-$r$ case as the vertex is not on the step boundary!
\end{remark}

%
%
\end{mdframed}

Since each edge is incident to the same number $\lfloor r/2\rfloor$ of incoming and outgoing active vertices on the boundary by definition of our matrices, it is straightforward to observe that we indeed pick up a total  of $2 \cdot \lfloor r/2\rfloor$ many $\sqrt{\ell}$ factors for each hyperedge across its first and last appearance. That is a total of $ 
\sqrt{\ell}^{2 \cdot \lfloor r/2\rfloor} = \ell^{\lfloor r/2\rfloor}
$ which is precisely the combinatorial factor incurred for identifying the outgoing active vertices for each edge when it appears for the first time!
\subsubsection{Factor of Incoming Active Vertices}

Next, we focus on the factor from incoming vertex for each edge that appears for the first time. In fact, our assignment will be more general that allows us to handle factor for edges making a middle (non-first/last) appearance and edges making last appearance but not in ideal-condition altogether , as they both involved factors in $O_r(q)$ as well.

\begin{mdframed}[frametitle = Combinatorial Factor Assignment for Vertex Factor in $n$ and $O(q)$ ] \label{prop:vtx-assignment}
 Each vertex requires a factor $n$ to be specified when it first appears in the walk, and a subsequent factor of $O(q)$ when it appears as an incoming active-vertex of some $F$ edge in the walk. We assign its corresponding factor  as the following,
\begin{enumerate}
	\item Assign a factor of $\sqrt{n } $ for the step if the vertex is appearing for the first time;
	\item Assign a factor of $\sqrt{n}$ for the step if the vertex is appearing for the last (final) time;
	\item Assign a factor of $O(q\cdot r )$ if the vertex is making a middle appearance as an \emph{incoming} active vertex.
	\end{enumerate}
Moreover, for completeness, for vertices incident to some $H$-edge, i.e. edge appearing for the middle time, or an $L$-edge but not in ideal-step condition, assign a total cost of $O_r(q)$ to all the (incidental) incoming active vertices.
\end{mdframed}

The crux of the above scheme is that we redistribute the factor of $n$ for a vertex incurred up front when it first appears among its first and last appearance, so that we attain the usual square-root saving that we anticipate for the spectral norm bound. 

\paragraph{Ideal last-cost, and Redistribution} With the bounds for vertices incident to some edge making the first appearances, we may now focus on the last-cost, the cost of specifying an edge appearing for the last time. We formally restate the ideal last-cost bound as highlighted in the overview, and show that it can be analogously extended to the odd-$r$ case in which we use a \emph{single} factor of $\ell$ to identify $2$ edges simultaneously in the analogous ideal-condition.

For convenience, we recall the condition for an ideal last-step.\begin{definition}[Ideal Last-Step for even-$r$]
	Suppose at some step-$t$ in the walk, we are walking from $S_t$ to $S_{t+1}$ along some edge that is promised to be making the last-appearance in the walk, we call the step from $S_t$ to $S_{t+1}$ an ideal last-step if \emph{some out-going vertex in $S_{t}\setminus S_{t+1} \subseteq [n]$ is appearing for the final time in the walk, i.e., it is not to-be-revisited upon the departure from $S_t$}.
\end{definition} 
We now re-state our observation earlier, that for any fixed vertex, if it is appearing for the last time in the boundary, its incident edge can be effortlessly specified. 

\begin{claim}[Final-Appearance Edges is fixed for any \emph{given} Final-Appearance Vertex] \label{claim:retrun-fixed-given-vertex}
Given any vertex with label in $[n]$ at step-$t$ via some edge making its final appearance, if the vertex is additionally making its final-appearance, the edge is fixed (i.e., there is a unique edge).
\end{claim}

\begin{proof} Throughout the walk, we can keep track of the edges that have appeared yet not made its final appearance. To do this, for each subsequent appearance of an edge, one may use a cost of $2$ to identify whether this is the final appearance.

	With this record, suppose there is any "outgoing"-active vertex making its final appearance in $U_t$, we observe that such candidate "final"-appearance edge is unique given the "outgoing"-active vertex.\end{proof}

\begin{corollary}
For even-$r$, from the $U_t$-boundary at step-$t$, a cost of $\ell$ is sufficient to specify a final-appearance edge provided some "out-going" vertex of the edge is making its final appearance.
\end{corollary}

Let's now extend the above to the odd-$r$ case in which each step now contains $2$-edges by definition of our Kikuchi matrix. We first modify the ideal last-step condition as the following.

\begin{definition}[Ideal Last-Step for odd-$r$] For any odd $r$,
	suppose at some step-$t$ in the walk, we are walking from $S_t$ to $S_{t+1}$ along some edges via some intermediate vertex $W_t$, we call the step from $S_t$ to $S_{t+1}$ an ideal last-step if
	\begin{enumerate}
		\item \emph{some out-going vertex in $S_{t}\setminus S_{t+1} \subseteq [n]$ is appearing for the final time in the walk, i.e., it is not to-be-revisited upon the departure from $S_t$};
		\item the intermediate vertex $W_t$ (incident to both edges) is appearing for the final time in the current step-$t$, i.e., it is not to be revisited once we arrive at $S_{t+1}$.
	\end{enumerate} \end{definition}

Note that for odd-$r$, the only distinction in the definition of an ideal last-step is that we additionally impose the constraint that the intermediate vertex $W_t$ is also making a final appearance. As shown below, this condition is crucial for us to identify both edges of a single-step via a single factor of $O_r(\ell)$.

\begin{proposition} For odd-$r$,
	a cost of $\ell \cdot r = O_r(1) \cdot \ell$ is sufficient to identify both edges in an ideal last-step.
\end{proposition}
\begin{proof}
	This follows by a double application of \cref{claim:retrun-fixed-given-vertex} as the following.
	Notice that by construction of our matrix, and thereby definition of our walks, the intermediate vertex is incident to both edges. Therefore, it suffices for us to specify the first edge via a cost of $\ell$ as in the even-$r$ case via applying \cref{claim:retrun-fixed-given-vertex} on some out-going vertex making its final appearance. 
	Once the edge is identified, the intermediate vertex can be specified again by a cost of $r$. Finally, apply \cref{claim:retrun-fixed-given-vertex} again but on the specified intermediate vertex identifies to us the second edge that is making its final appearance. This incurs a total cost of $r\cdot \ell = O_r(\ell)$. 
\end{proof}

Finally, analogously to splitting the factor of $q = \ell \log n$  equally to the first and last step in which an edge appears, we redistribute the factor of $O_r(\ell)$ from the ideal last-step. Towards this end, we first observe the following bound for number of ideal steps,

\begin{proposition}
	There can be at most $q$ ideal last-steps in a length-$2q$ walk.
\end{proposition}
 \begin{proof}
 	This bound is immediate for even-$r$ as each edge needs to appear at least twice in the walk, therefore, at most half of the steps can be last-steps. To extend to the odd case, observe that we would have used at least $2q+2$ distinct edges if there are $\geq q+1$ ideal last steps. However, we use at most $4q$ edges (counting multiplicities) in a length-$2q$ walk, and each distinct edge needs to appear at least twice, leading to a contradiction.
 	 \end{proof}
 We may now "naively" split the factor of $\ell$ as following. It should be emphasized that the cost of $O_r(\sqrt{\ell})$ is assigned to the whole step as opposed to an edge alone. This does not make a distinction for the even-$r$ case, while importantly for odd-$r$, each step of \emph{two edges} gets assigned a total cost of $O_r(\sqrt{\ell})$.
 
\begin{mdframed}[frametitle = Last-Step Cost Assignment] \label{prop: last-cost-assignment}
For each step that uses some edge making the first or final appearance $(F/L)$ , we assign a factor $O_r(\sqrt{\ell})$ to the whole-step.
\end{mdframed}

Finally, we wrap up this section by recapping our factor assignments and verify that this is a complete factor-assignment scheme for the combinatorial factor. 
\begin{proposition}
    Let $B_{global}(P)$ be the sufficient cost to identify the whole walk, and $B_{local}(P)$ be the cost assigned by our scheme to each step across the walk, \[ 
    B_{global}(P) \geq B_{local}(P)\,,
    \]
    for \[ 
    B_{local}(P)\coloneqq \text{MatrixDimension} \cdot  \prod_{i}B_{local}(i)\,.
    \]
    In other words, each combinatorial factor for counting the walk is assigned to some local step. 
\end{proposition}
\begin{proof}
    To start with, notice that we use a cost of $n^\ell$ to identify the walk-start in both the global and local accounting. Next, we prove the the combinatorial factor regarding edges that make the first appearance, the cost in $\ell$ for specifying the out-going active vertex is accounted in \cref{prop:out-going-cost}, and the cost of specifying incoming active vertices is accounted by  lemma~\cref{lem:first-appearance-accouning}. For combinatorial cost regarding edges that make $H$ appearance, or non-ideal $L$-step, we assign a factor of $O_r(q)$ to each such edge in \cref{prop:vtx-assignment}. Finally, for each edge making last appearance in the step in ideal-step, a total factor of $\ell$ is sufficient and we pick up two $\sqrt{\ell}$-factors from the first and last appearance according to \cref{prop: last-cost-assignment}.
    
\end{proof}

\begin{lemma} \label{lem:first-appearance-accouning}
	This is a complete vertex-assignment scheme regarding combinatorial factor of edges appearing for the first time. In other words, for any vertex $v$, let $B_{F}(v)$ be the total factor assigned to this vertex throughout the walk across its various appearances as an incoming active vertex of some edge appearing for the first time, we have \[ 
	B_{F}(v) \geq 	(n  )\cdot  ( 2q \cdot r)^{s(v)} 
	\]
	where we recall that $s(v)$ is the number of steps in which $v$ is specified as an incoming vertex of some edge beyond the vertex's first appearance. 
\end{lemma}
\begin{proof}
    To start with, notice the RHS is the combinatorial factor incurred by $F$ edges that arise from specifying the label of $v$. In the case $s(v) = 0$ and the vertex appears only twice in the walk via first and last appearance, we have $B_{vtx}(v) = (\sqrt{n  })^2 = n$. For $s(v) \neq 0$, notice we still pick up $n$ from the first and last appearance $v$ which offset the contribution of $n$ on the RHS. For the factor depending on $s(v)$, consider
	 each time $v$ appears as an incoming active edge beyond the first time, we assign to its particular appearance step a factor of $2q\cdot  r $, which matches the cost assigned to vertex $v$ via that particular edge. Taking product over all edges that contribute to $s(v)$ gives the desired. 

\end{proof}

\subsection{Step-Bound from Factor Assignments}
We are now ready to combine the above components and deduce our local bound in the factor assignment scheme. Before that, we observe the following meta-claim about vertex-appearance that allows us to exploit the evenness assumption of our walks such that each edge (i.e. random variable) needs to appear at least twice.
\begin{claim}\label{claim: meta-appearance}
	In the case the edge is via an $F$ (or $L$)-step, any incidental out-going vertex (or any incoming vertex)  cannot be making \emph{last} (or \emph{first} appearance). 
\end{claim}
\begin{proof}
	This follows from the observation that edge needs to appear at least twice throughout the walk. Observe that it suffices for us to consider active vertices depending on whether they are outgoing or incoming. For an $F$-step, if any "outgoing" active vertex makes the final appearance at step $t$, the edge used by step-$t$ would only appear only once as otherwise any active vertex of this edge would make a subsequent appearance. 
	
	Similarly, for an $L$-step while any incoming vertex is making a first appearance, the edge at step-$t$ would have appeared only once throughout the walk.
\end{proof}

\begin{lemma}[Step Bound for Even-$r$]  For some fixed constant $\delta>0$ and let  $B_q(\al)$ be cost assigned to step-$t$ with step-status $\al \in \{F,H,L\} $, we have \[ 
	B_q(F) = B_q(L) \leq O_r(1) \cdot  \left(\sqrt{n\cdot \ell}\right)^{r/2} \cdot \sqrt{\ell} \,,
	\]
	and \[ 
	B_q(H) \leq o_n(1)  \left(\sqrt{n\cdot \ell}\right)^{r/2}\,.	\]
	provided $q \ll n^\delta$.
\end{lemma}

\begin{proof}
We case on the step-status from $\{F,H,L\}$. For notational convenience, consider the step-$t$ move from boundary set $U_t = S_t$ to $V_t = S_{t+1}$.
	In the case this is an $F$-step, we observe the following, \begin{enumerate}
		\item Each dormant vertex in the intersection of $U_t\cap V_t$ not contribute any factor;
		\item By \cref{claim: meta-appearance}, any outgoing vertex in $U_t\setminus V_t$ cannot be making its last appearance. Moreover, since each such vertex appears in the previous step and thus cannot be making their first appearance. Finally, note that none of such vertex may be specified via the middle appearance cost for "incoming" active vertex, hence there is no cost for any such vertex.
		\item Any incoming vertex in $V_t\setminus U_t$ cannot be making its last appearance, while it can be making its first or middle appearance. In the case this is a first-appearance, we assign it a cost of $\sqrt{n}$; and for any subsequent middle appearance, we assign a cost of $2q \cdot r$ via \cref{prop:vtx-assignment}, 
		\item Additionally, any incoming vertex in $V_t\setminus U_t$ is assigned a factor of $\sqrt{\ell}$ via \cref{prop:out-going-cost} for the "split" cost of specifying the out-going active vertex of the edge;
		\item Therefore, each vertex contributes a cost of \[ 
		  (\sqrt{n} +  \ell \cdot 2q \cdot r) \sqrt{\ell} = O_r(1)\cdot  \sqrt{n\ell} \,.
		\]
        \item Edge-value is $1$ for each $F$ step via \cref{prop: edge-value-assignment};
        \item Each $F$ step gets assigned a factor of $O_r(\sqrt{\ell})$ from (potential) ideal last-step cost via \cref{prop: last-cost-assignment};
		\item Combining the cost for all vertices in this step gives us a bound of \[ 
		B_q(F)\leq  O_r(1) \cdot  (  \sqrt{n\ell} )^{r/2 }	 \cdot \sqrt{\ell}
		 \] as there are at most $r/2$ vertices in $V\setminus U$.
	\end{enumerate}
	
	In the case this is an $L$-step, the factor is essentially the same as the $F$-step except vertices in $U$ may now contribute as last-appearance vertex (as opposed to those in $V$ contributing as first-appearance vertex). Formally, we have \begin{enumerate}
		\item Each dormant vertex in the intersection of $U_t\cap V_t$ not contribute any factor;
		\item By \cref{claim: meta-appearance}, any incoming vertex in $V_t\setminus U_t$ cannot be making its first appearance. Moreover, since each such vertex appears in the subsequent step, and thus cannot be making their last appearance. Finally, note that none of such vertex may be specified via the middle appearance cost for "incoming" active vertex since such a cost is only incurred for edges appearing for the first time. Therefore, no cost is assigned for any vertex in $V_t\setminus U_t$.
		\item Any outgoing vertex in $U_t\setminus V_t$ cannot be making its first appearance, while it can be making its last or middle appearance. In the case this is a last-appearance, we assign it a cost of $\sqrt{n}$; and for any subsequent middle appearance, we assign a cost of $2q \cdot r$  via \cref{prop:vtx-assignment}. 
		\item Any outgoing vertex in $U_t\setminus V_t$ is on the step-boundary of an edge appearing for the last time, and therefore assigned a factor of $\sqrt{\ell}$ via \cref{prop:out-going-cost}  for the "split" cost of specifying the out-going active vertex of the edge;
		\item Therefore, each vertex  $U_t\setminus V_t$  contributes a cost of \[ 
		  (\sqrt{n} +  2q \cdot r)\cdot \sqrt{\ell } = O_r(1)\cdot  \sqrt{n\ell} \,.
		\]
		\item Edge-value is $1$ for each $L$ step via \cref{prop: edge-value-assignment};
        \item Each $L$ step gets assigned a factor of $O_r(\sqrt{\ell})$ from (potential) ideal last-step cost via \cref{prop: last-cost-assignment};
		\item Combining the above, notice that the factor of $2q$ is only needed if all vertices from  $U\setminus V$ are making middle appearance, in which case we do not pick up any final-appearance $\sqrt{n\ell }$ factor. Therefore, we have the total bound of \[ 
		B_q(L)\leq  \sqrt{\ell} \cdot ((1+o_n(1) \sqrt{n\ell })^{r/2} + \sqrt{\ell}\cdot  2q \cdot  (\sqrt{n\ell })^{r/2 -1}  = \ O_r(1)  \cdot (\sqrt{n\ell })^{r/2}  \cdot \sqrt{\ell}  \,.
		\]	
		provided $q \leq n^\delta$ for some $\delta>0$. 
		\end{enumerate}
		
		Finally, for the case of $H$-step, it is immediate to observe that all vertices must be making a middle appearance with no extra vertex-cost as the edge can be specified at a cost of $2q$. Combining with the edge-value gives us a bound of \[ 
		B_q(H)\leq \underbrace{(2q)}_{\text{edge-val}} \cdot \underbrace{(2q)}_{\text{combinatorial factor} } =   (2q)^{2} \,.
		\]

\end{proof}

\begin{lemma}[Step Bound for Odd-$r$]  For some fixed constant $\delta>0$, for any $q<n^\delta$,
	let $B_q(\al)$ be cost assigned to step-$t$ with step-status $\al \in \{F,H,L\} \times \{F, H, L\} $, we have \[ 
	B_q(F\times  F)  = B_q(L \times L) \leq O_r(1) \left(\sqrt{n\cdot \ell}\right)^{r}  \,,
	\]
	and \[ 
	B_q(\beta) \leq o_n(1)  \left(\sqrt{n\cdot \ell}\right)^{r}\,.	\]
	for any step-status $\beta \in \{F,H,L\} \times \{F, H, L\} \setminus \{F\times F, L \times L\}$. 
	In other words, we treat any step other than both edges being simultaneously $F$ or $L$ as a lower-order term. 
\end{lemma}
\begin{proof}
	The proof is largely identical to the even-$r$ case for the per-vertex factor of each active vertex except that the number of active vertices differs, except the factor of $\ell$ from outgoing active vertices of an $F$-edge warrants extra care as discussed in \cref{prop:out-going-cost}  .
	
	In the case of $F\times F$ status, notice the following change, \begin{enumerate}
		\item There are at most $r$ incoming active vertices that could be making first appearance, so we get assigned a factor of at most $O_r(1) \cdot \sqrt{n}^r$
		\item There are $r-1$ incoming active vertices on the boundary, each of which gets an assigned factor of $\sqrt{\ell}$ from  \cref{prop:out-going-cost}		\item Each step gets assigned a total of $O_r(\sqrt{\ell})$ from \cref{prop: last-cost-assignment}  .
	\end{enumerate} 
	Combining the above gives us  $B_q(F\times F) \leq  O_r(1) \left(\sqrt{n }\cdot \sqrt{\ell} \right)^r $.
	
	 The case of $L\times L$ status also warrants attention on its own,
	 \begin{enumerate}
	 	\item There are at most $r$ outgoingg active vertices that could be making last appearance, so we get assigned a factor of at most $O_r(1) \cdot \sqrt{n}^r$
		\item There are $r-1$ outgoing active vertices on the boundary, each of which gets an assigned factor of $\sqrt{\ell}$ from \cref{prop:out-going-cost} ;
		\item Each step gets assigned a total of $O_r(\sqrt{\ell})$ from \cref{prop: last-cost-assignment}  .
	 \end{enumerate}
	 Combining the above gives us  $B_q(L\times L) \leq O_r(1) \left(\sqrt{n }\cdot \sqrt{\ell} \right)^r $  .
	The bounds for other step-status are lower-order term, and follow immediately from the  above discussion.
\end{proof}

	\paragraph{Wrapping Up}
	We are now ready to prove our main theorem. 
	\begin{proof}[Proof to \cref{thm:main-thm-even} and  \cref{thm:main-thm-odd}]
	 We focus on the even-$r$ case while the odd-$r$ case holds verbatim. Let $B_q$ be our final step-value bound for each step by our factor-assignment scheme, summing over all possible step-status, we have \[ 
	 B_q \leq B_q(F) + B_q(H) + B_q(L) \leq O_r(1) \sqrt{n\cdot \ell }^{r/2} \cdot \sqrt{\ell }
	 \]
	 By construction of our factor-assignment scheme \cref{def:step-bound-function}, this gives an upper bound of \[ 
	 \E[\Tr(M_\ell)^{2q}] \leq \text{MatrixDimension} \cdot (B_q)^{2q}\,.
	 \]
	 Finally, this translates to a matrix norm bound immediately by Markov's as we consider for any constant $\eps >0$. \begin{align*}
	 	\Pr[ \|M_\ell\|_{sp} \geq  (1+\eps)\cdot B_q  ]&\leq \frac{\E[\Tr(M_\ell)^{2q}] }{(1+\eps)^{2q} \cdot (B_q)^{2q}  }\\
	 	&\leq \frac{n^\ell \cdot (B_q)^{2q}}{(1+\eps)^{2q} \cdot (B_q)^{2q}}\\
	 	&\leq c^{-q/\log n}
 	 \end{align*}
 	 for some constant $c>0$  since
 	 our $q$ can be taken as $q<n^\delta$ for some constant $\delta>0$.	
	\end{proof}

\section{Lower Bound}
We now complement our upper bound result by a lower bound. For this section, we restrict our attention to the even-$r$ case and note that an analogous argument extends to odd-$r$ case.
\label{sec:lower-bound}
\begin{lemma}
	\begin{align*}
		\E[\Tr(M_\ell^{2q})]& \geq  \text{matrix-dimension}\cdot \left((2\ell/r)^{r/4 }  \cdot (2\ell /r) \cdot \sqrt{\binom{n-\ell}{r/2}}\right)^{2q}\\
		&=   \text{matrix-dimension}\cdot  \left(\Theta_r(1) \cdot  \sqrt{n^{r/2} \cdot \ell^{r/2}} \cdot \sqrt{\ell}\right)^{2q} \,.
	\end{align*}
\end{lemma}
\begin{proof}
	
%
%

To give a lower bound, we observe that each term in the expected trace is non-negative, and therefore it suffices for us to focus on a specific type of walks in the Kikuchi graph (such that it uses each edge twice so that it has expected-value $1$), and then show that the combinatorial count of such walks is large. Consider the walks of the following type,
\begin{enumerate}
	\item It would be useful to consider $m \coloneqq \ell /0.5r = 2\ell/r$ buckets, and  consider splitting the length-$2q$ walk into a walk of $2q/(2m)$ chunks, with each chunk being a walk of length-$2m$.
	\item For each chunk of length-$2m$, this is a walk consist of $m$ many $F$-steps and $m$ many $L$ steps (moreover, we focus on $F$-step leading to "new" vertices only).
	\item For each chunk, there are  $\binom{\ell}{r/2, r/2, r/2 ...., r/2 }$ ways to bucket $r/2$ vertices among $\ell$ vertices, and each bucket will get to walk twice, once via an $F$-step and the other via an $L$-step.
	\item Observe that we have $\binom{2m }{2,2,2\dots,2}$ many orders to pick which bucket "walks" at each step, and regardless of the ordering, we are guaranteed to return to the start at the end of this chunk. In other words, each arbitrary ordering gives a valid walk.
	\item For each $F$-step we have the usual factor of $\binom{n-\ell}{r/2}$, while notice we no longer have the factor of $\binom{\ell}{r/2}$ since the bucket has been determined.
		\end{enumerate}
 Next, we recall the following fact about the center coefficient of multinomials,
 \begin{fact}
 There exists some constant $c_k>0$ such that
 	\[ 
 	\log \binom{kn}{n} \geq  c_k \cdot  kn\log k
 	\]
 	for sufficiently large $n$.
 \end{fact}
  Thus, we have a total count of the length-$2m$ chunk as \[ 
  \binom{\ell}{r/2, r/2, r/2 ...., r/2 } \cdot \binom{2m }{2,2,2\dots,2} =\exp(\ell \log (2\ell/r)) \cdot \exp(2m \log m) \cdot \binom{n-\ell}{r/2}^m
  \]
  Taking the $2m$-th root gives us the cost per step as \[
  \Theta(1) \cdot  2^{\frac{\ell}{2m }  \cdot \log (2\ell/r) } \cdot m \cdot \sqrt{\binom{n-\ell}{r/2}}\,,
   \]
   Recall that $m=2\ell/r$, we have \[ 
   \Theta(1) \cdot  (2\ell/r)^{r/4 } \cdot (2\ell /r) \cdot \sqrt{\binom{n-\ell}{r/2}}\,.
   \]
   Recall that our upper bound reads as (per step) \[ 
   O_r(1) \cdot  \sqrt{n^{r/2} \cdot \ell^{r/2+1}} \,,
   \]
   this is tight up to the hidden constant in $O_r(1)$.
\end{proof}
Next, analogously to lower bounding spectral norm of a Wigner matrix,  to deduce the concentration of $\|M_\ell\|$ within $\Theta_r(1) \cdot  \sqrt{n^{r/2} \cdot \ell^{r/2}} \cdot \sqrt{\ell}$, we use the usual relation between the spectral norm and trace to lower bound $ \E[ \|M_\ell\|^{2q}]$. For simplicity, we focus on the setting when $G$ is a random symmetric $\{\pm 1\}$ tensor, while note that the analogous bound can be extended to Gaussian via more complicated usage of Gaussian concentration inequality.  Crucially , since $\|M_\ell\|_{sp}$ is sub-gaussian by \cref{prop:variance-bound} and \cref{thm:talagrand}, we can further deduce
\begin{corollary}
	For $G$ an random symmetric tensor of Bernoulli input,
   \[ 
   \E[||M_\ell(G)\|_{sp}] \geq \Theta_r(1) \cdot  \sqrt{n^{r/2}\cdot \ell^{r/2}}\cdot \sqrt{\ell}\,,
   \]
   and moreover, with probability at least $1-o_n(1)$, \[ 
   \|M_\ell(G)\|_{sp} \geq \Theta_r(1) \cdot  \sqrt{n^{r/2}\cdot \ell^{r/2}}\cdot \sqrt{\ell}\,.
   \]
\end{corollary}

We defer the calculation for high-moment-to-mean to the appendix while noting that it is analogous to the Lipschitz concentration of random matrix of i.i.d. entries.

\section*{Acknowledgements}
The authors thank Afonso Bandeira and Petar Nizic-Nikolac for multiple illuminating discussions that posed and discussed this problem and for sharing an early version of their manuscript~\cite{bandeira_nizic_kikuchi_pc} with us. We would also like to thank Tim Hsieh for the insightful discussions, and the anonymous reviewers from SODA'26 for various helpful comments and suggestions.

\clearpage\newpage
\bibliographystyle{alpha}
\bibliography{bib}
\clearpage\newpage
\section*{Appendix}
\subsection{Deferred Proofs}
\begin{claim}[Detection from Spectarl Norm Bound: even-$r$] For even $r>3$,
		for a given upper bound $B(G)$ on a random tensor $G$, there is a spectral algorithm that solves the detection question for \[\lambda \geq \Theta(1) \cdot \frac{B(G)}{\binom{n-\ell}{r/2} \cdot \binom{\ell}{r/2}  }  = \Theta_r(1) \cdot  \frac{B(G)}{n^{r/2} \cdot \ell^{r/2}}  \,.\]
\end{claim}

\begin{proof}
	Consider the tensor with planted spike $\widetilde{T} = \lambda \cdot v^{\otimes r} + G$, we have \[ 
	\|M_\ell(\widetilde{T}) \|_{sp} \geq  \|M_{\ell}( \lambda \cdot v^{\otimes r}  )\|_{sp} - \|M_{\ell}(G)\|_{sp}  \geq  \|M_{\ell}( \lambda \cdot v^{\otimes r}  )\| - B(G)\,.
	\]
	Moreover, consider the vector $\tilde{v}\in \R^{\binom{n}{\ell}}$ with entries $\tilde{v}[S]  \coloneqq v_S = \prod_{i\in S} v[i] $. By Spectral Theorem, for even $r$, we have \begin{align*}
		\|M_{\ell}( \lambda \cdot v^{\otimes r}  )\|_{sp}  \geq \frac{1}{\|\tilde{v}\|_2^2} \cdot   \tilde{v}^T \cdot M_\ell( \lambda \cdot v^{\otimes r}  )  \cdot \tilde{v}^T  &= \frac{1}{\|\tilde{v}\|_2^2} \cdot \sum_{\substack{I, J \in \binom{n}{\ell} \\ |I\Delta J | = r }} (\lambda\cdot  v^{\otimes r}_{I\Delta J}) \cdot v_I \cdot v_J \\
		&= \frac{1}{\|\tilde{v}\|_2^2} \cdot \sum_{\substack{I, J \in \binom{n}{\ell} \\ |I\Delta J | = r }} \lambda\cdot v_{I\Delta J}^2 \cdot v_{I\cap J}^2\\
		&= \lambda \cdot \frac{\binom{n}{\ell} \cdot \binom{n-\ell}{r/2}\cdot  \binom{\ell}{r/2} }{\binom{n}{\ell}}\\
		&= \lambda \cdot  \binom{n-\ell}{r/2}\cdot  \binom{\ell}{r/2}  \,,
		  	\end{align*}
		  	where the first equality follows from the expansion of quadratic form and $M_\ell(\lambda \cdot v^{\otimes r}) [I,J] = \lambda\cdot v^{\otimes r}_{I\Delta J} $, and the second-to-last equality from $v_i^2 = 1$ for any $i\in[n]$ since $v$ is a boolean vector and observing that for each fixed $I\subseteq [n]$, there are $\binom{n-\ell}{r/2}\cdot \binom{\ell}{r/2}$ many adjacent indices $J \in \binom{n}{\ell} $ such that $|I\Delta J| = r$.

	
	On the other hand, for an (unplanted) random symmetric tensor $T=G$, we have $\|M_\ell(T)\|_{sp} \leq B(G)$ by assumption. Thus, we have a separation as long as \[ 
	\|M_\ell(\widetilde{T})\|_{sp} > B(G) \geq \|M_{\ell}(T)\|_{sp} \,,
	\]
	which happens as long as $\lambda \cdot  \binom{n-\ell}{r/2}\cdot  \binom{\ell}{r/2} > 2 \cdot B(G) $. Rearranging then gives us the desired.
\end{proof}
Analogously, we have the following for odd-$r$. 
\begin{claim}
    For odd $r\geq 3$,
		for a given upper bound $B(G)$ on a random tensor $G$, there is a spectral algorithm that solves the detection question for \[\lambda \geq \ \Theta_r(1) \cdot  \left( \frac{B(G)}{n^{r+1} \cdot \ell^{r}}  \right)^{1/2} \,.\]
\end{claim}
\begin{proof}
    The proof is essentially verbatim, except one verify that for the odd Kikuchi matrix,    we have \begin{align*}
         \|M_{\ell}( \lambda \cdot v^{\otimes r}  )\|_{sp} & \geq \frac{1}{\|\tilde{v}\|_2^2} \cdot   \tilde{v}^T \cdot M_\ell( \lambda \cdot v^{\otimes r}  )  \cdot \tilde{v}^T  \\&= \frac{1}{\|\tilde{v}\|_2^2} \cdot \sum_{\substack{I\in \binom{n}{\ell}  }}  \sum_{\substack{ A,B\subseteq I \\ |A| =|B|=\lfloor r/2 \rfloor \\ A\cap B = \emptyset } } \sum_{t\in [n]  } \sum_{\substack{ S, T \in \binom{n}{\lfloor r/2 \rfloor } \\S\cap T = \emptyset \\ S\cap I =\emptyset\\ T\cap I =\emptyset \\ J = I\setminus (A\cup B) \cup S\cup T }  } v_{A\cup S\cup \{t\}}  \cdot v_{B\cup T\cup \{t\}}  \cdot \cdot v_I \cdot v_{J} \cdot \lambda^2 \\
         &= \frac{1}{\|\tilde{v}\|_2^2} \cdot \sum_{\substack{I\in \binom{n}{\ell}  }}  \sum_{\substack{ A,B\subseteq I \\ |A| =|B|=\lfloor r/2 \rfloor \\ A\cap B = \emptyset } } \sum_{t\in [n]  } \sum_{\substack{ S, T \in \binom{n}{\lfloor r/2 \rfloor } \\S\cap T = \emptyset \\ S\cap I =\emptyset\\ T\cap I =\emptyset \\ J = I\setminus (A\cup B) \cup S\cup T }  } v_I^2 \cdot  v_J^2 \cdot  v_t^2 \cdot \lambda^2\\
         &= \Theta_r(1) \cdot  \lambda^2 \cdot  n^{r+1} \cdot \ell^r \,.
     \end{align*}       
\end{proof}
The above two proofs also immediately allow us to prove our main theorem for algorithmic application for smooth trade-off of Tensor PCA,
\begin{proof}[Proof to~\cref{thm:tensor-pca-application}]
    This follows by plugging in our established bound from the main lemma for norm bounds \cref{thm:main-thm-even}, \cref{thm:main-thm-odd} into the above bound that translates a high-probability-bound on the Kikuchi matrix to the detection threshold $\lambda$.
\end{proof}

\subsection{Extension to Recovery}
\begin{proof}[Proof Sketch to \cref{thm:recovery} ]
    Our extension result to the recovery question follows verbatim from that in A.3 of \cite{WAM19} except that we apply the improved spectral norm bound for the noise matrix in the analysis of the voting matrix.  We highlight the change from our spectral norm bound, that is
   \begin{lemma}[Lemma A.5+ A.6 of \cite{WAM19}]  
   	 For any $\eps>0$, and $\delta\in (0,1)$, for $\lambda > \epsilon^{-1} \cdot c \cdot \sqrt{ \frac{C_r}{d_\ell } } $, then with probability at least $1-\delta$, \[ 
   	 \|v^\perp\|_2^2 \leq \frac{\|M_\ell\|}{\lambda \cdot \binom{n}{r/2} \cdot \binom{n}{\ell/2} } \cdot \frac{\ell}{m}
   	 \]
   	 \end{lemma}
    The subsequent analysis follows from the remaining of section A from \cite{WAM19}.
\end{proof}

\subsection{Deferred Proofs for Lower Bound}
We recall the following bounded difference inequalities,
   \begin{fact}
       For any function $f(X_1,...,X_n)$ of independent random variables $\{X_1,...,X_n\}$, we have \[ 
       \Var[f(X_1,...,X_n)] \leq \E[\sum_{i}^n (D_i^- f(X_1,...,X_n)^2 ]
       \]
       for \[ 
       D_i^{-} f(X) \coloneqq f(X_1,...,X_n) - \inf_z (X_1,..,X_{i-1},z, X_{i+1},...,X_n)\,.
       \]
   \end{fact}
   
   \begin{theorem}[Theorem 4.20 in \cite{vanHandel2016})]\label{thm:talagrand}  In the above set-up, $f(X_1,...,X_n)$ is $\|\sum_i (D_i^- f)^2\|_\infty $ -subgaussian.    
   \end{theorem}
From this point, it suffices for us to verify the variance bound of $1$ for the Kikuchi matrix in our setting and apply the theorem. 
   
   \begin{proposition} \label{prop:variance-bound}
       \[ 
        \Var[f(G_1,...,G_{\binom{n}{r} } )] \leq O(1)\,.  
       \]
   \end{proposition}
   \begin{proof}
       We highlight our argument by considering random Bernoulli matrix. Consider for the time that we have fixed underlying input $G$, and we will vary the entry of $G$ along the way. For any $S\in \binom{n}{r}$, let $G_S'$ be the tensor obtained by varying entry in $G[S]$ while fixing the rest of the entries. 
     For the following, let $M_\ell \coloneqq M_\ell(G)$ be the Kikuchi matrix of the fixed input $G$, and  
       let $M^-$ be the Kikuchi matrix obtained  as the following, \[ 
       \lambda_{\max}(M^-) = \inf_{G' \leftarrow G_S \in \{-1, 1\} } \lambda_{max}(M_\ell(G'))
       \]
       Let $v_{\max}$ be the top eigenvector for $M_\ell \coloneqq M_\ell(G)$ of unit-norm. Observe that \begin{align*}
           D_S^- \lambda_{\max}(M_\ell  ) &= \lambda_{\max}(M_\ell ) - \lambda_{\max}(M^-)\\
           &\leq \langle v_{\max}(M_\ell ) v_{\max} \rangle - \sup_{v\in B_2} \langle v,M^{-} v\rangle    \\
           &\leq  \langle v_{\max}(M_\ell - M^- ) v_{\max} \rangle   \\         
           &=2  \sum_{\substack{ A_S,B_S \in \binom{S}{r/2} \\ \text{equal partition for }S  } } \sum_{D\in \binom{[s]\setminus S}{\ell - r/2}  } v_{\max}[A_S\cup D]\cdot v_{\max}[B_S\cup D]  \,.
       \end{align*}  
       Taking the square of the above, we have \begin{align*}
      \sum_{\substack{ A_S,B_S, A_S', B_S' \in \binom{S}{r/2} \\ \text{equal partition for }S  } } \sum_{D, D'\in \binom{[s]\setminus S}{\ell - r/2}  } v_{\max}[A_S\cup D]\cdot v_{\max}[B_S\cup D] \cdot  v_{\max}[A_S'\cup D']\cdot v_{\max}[B_S'\cup D']
       \end{align*}
       Finally, summing over $S$, we have \begin{align*}
       	\sum_S (D_S^-)^2 \leq O(1) \sum_{T \in \binom{n}{\ell}} v_{\max}[T]^2  = O(1)
       \end{align*}
   \end{proof}

\end{document}